\documentclass[10pt,journal,compsoc]{IEEEtran}
\usepackage{xcolor}
\usepackage[nocompress]{cite}
\usepackage{graphicx}
\usepackage{amsmath}
\usepackage{epsfig}
\usepackage[noend]{algorithmic}
\usepackage[caption=false, font=footnotesize]{subfig}
\usepackage{multirow}
\usepackage{tabularx}
\usepackage{enumitem}
\usepackage[mathscr]{eucal}

\usepackage{url}

\usepackage[linesnumbered,ruled,algonl,vlined,noend]{algorithm2e}

\SetKwRepeat{Do}{do}{while}

\usepackage{amsthm} 
\newtheorem{lemma}{Lemma}

\newtheorem{example}{Example}

\theoremstyle{definition}
\newtheorem{definition}{Definition}

\usepackage[noend]{algorithmic}

\newcommand{\algorithmiccontinue}{\textbf{continue}}
\newcommand{\CONTINUE}{\STATE \algorithmiccontinue}

\renewcommand{\algorithmiccomment}[1]{\bgroup\hfill//~#1\egroup}

\newcommand{\hui}[1]{{\color{red}{\bf{Hui says:}} \emph{#1}}}

\newcommand{\shane}[1]{{\color{orange}{\bf{Shane says:}} \emph{#1}}}

\newcommand{\myparagraph}[1]{\vspace{0.45\baselineskip}\noindent{\textbf{#1.}}~}

\newcommand{\rider}{\ensuremath{r_{j}}}
\newcommand{\riderone}{\ensuremath{r_1}}

\newcommand{\ridertwo}{\ensuremath{r_2}}
\newcommand{\driver}{\ensuremath{d_{i}}}
\newcommand{\tripschedule}{\ensuremath{tr}}
\newcommand{\ntripschedule}{\ensuremath{tr^{\prime}}}

\newcommand{\augdistance}{\ensuremath{\mathscr{U}}}

\newcommand{\optaugdistance}{\ensuremath{\mathscr{U}^{*}}}

\newcommand{\algoDistancefirst}{{\textit{Distance-first}}\xspace}
\newcommand{\algoGreedy}{{\textit{Greedy}}\xspace}

\newcommand{\partitionnum}{\ensuremath{\tau}}

\newcommand{\source}{\ensuremath{l^{s}}}
\newcommand{\des}{\ensuremath{l^{d}}}

\newcommand{\dsa}{\ensuremath{DSA}\xspace}
\newcommand{\df}{\ensuremath{DF}\xspace}
\newcommand{\dfp}{\ensuremath{DF}+\ensuremath{P}\xspace}
\newcommand{\gr}{\ensuremath{GR}\xspace}
\newcommand{\grp}{\ensuremath{GR}+\ensuremath{P}\xspace}
\newcommand{\sa}{\ensuremath{SA}\xspace}

\newcommand{\sriders}{\ensuremath{R_s}}

\newcommand{\detourdis}{\ensuremath{\triangle d}}

\newcommand{\ridersource}{\ensuremath{l^{s}}}
\newcommand{\riderdes}{\ensuremath{l^{d}}}

\newcommand{\shortestpathquery}{\ensuremath{\Omega}} 
\newcommand{\detour}{\ensuremath{\theta}} 

\newcommand{\wait}{\ensuremath{w}\xspace} 

\newcommand{\ms}{\ensuremath{c}} 

\newcommand{\updatetimeinter}{\ensuremath{\triangle t}\xspace}
\newcommand{\loc}{\ensuremath{l}\xspace} 
\newcommand{\filter}{{\textit{Filter}}\xspace}
\newcommand{\refine}{{\textit{Refine}}\xspace}

\newcommand{\utilitydf}{\ensuremath{U_{df}}\xspace} 
\newcommand{\utilitygr}{\ensuremath{U_{gr}}\xspace} 

\newcommand{\capacity}{\ensuremath{cp[k]}\xspace}
\newcommand{\slack}{\ensuremath{sd[k]}\xspace}

\newcommand{\additiondis}{\ensuremath{AD}\xspace}

\newcommand{\speed}{\ensuremath{sp}\xspace}

\newcommand{\dist}{\ensuremath{dis}}

\newcommand{\distactual}{\ensuremath{dis_a}}
\newcommand{\distsd}{\ensuremath{\distactual(\source,l^d)}}
\newcommand{\distminsd}{\ensuremath{dis(\source,l^d)}}

\newcommand{\timesd}{\ensuremath{t_a(\source,l^d)}}
\newcommand{\timeminsd}{\ensuremath{t(\source,l^d)}}

\begin{document}
	\title{Dynamic Ridesharing in Peak Travel Periods}
	
	\author{Hui~Luo,~Zhifeng~Bao,~Farhana~M.~Choudhury,~and~J.~Shane~Culpepper
		\IEEEcompsocitemizethanks{
			\IEEEcompsocthanksitem H. Luo, Z. Bao (corresponding author) and J. Culpepper are with the School of Science, Computer Science, and Information Technology, RMIT University, Melbourne, VIC 3000, Australia.\protect\\
			E-mail: firstname.surname@rmit.edu.au
			\protect
			\IEEEcompsocthanksitem F. Choudhury is with the School of Computing and Information Systems, The University of Melbourne, Melbourne, VIC 3000, Australia. \protect\\
			E-mail: farhana.choudhury@unimelb.edu.au
			}
	}


	\IEEEtitleabstractindextext{%
		\begin{abstract}
			The abstract goes here.
		\end{abstract}
		
		}

	\IEEEtitleabstractindextext{%
	\begin{abstract}
In this paper, we propose and study a variant of the dynamic
ridesharing problem with a specific focus on peak hours: Given a set
of drivers and a set of rider requests, we aim to match drivers to
each rider request by achieving two objectives:
maximizing the served rate and minimizing the total additional distance, subject to a series of spatio-temporal constraints.
Our problem can be distinguished from existing ridesharing solutions
in three aspects: (1) Previous work did not fully explore the impact
of peak travel periods where the number of rider requests
is much greater than the number of available drivers.
(2) Existing ridesharing solutions usually rely on single objective
optimization techniques, such as minimizing the total travel cost
(either distance or time).
(3) When evaluating the overall system performance, the runtime spent on updating drivers' trip schedules as per newly coming rider requests should be incorporated, while it is unfortunately excluded by most existing solutions.  
In order
to achieve our goal, we propose an underlying index structure 
on top of a partitioned road network, and compute the lower bounds of the
shortest path distance between any two vertices.
Using the proposed index together with a set of new pruning rules, we
develop an efficient algorithm to dynamically include new riders
directly into an existing trip schedule of a driver.
In order to respond to new rider requests more effectively, we
propose two algorithms that bilaterally match drivers with rider
requests.
Finally, we perform extensive experiments on a large-scale test
collection to validate the effectiveness and efficiency of the
proposed methods.
\end{abstract}

}

	\maketitle
	
	\IEEEdisplaynontitleabstractindextext
	\IEEEpeerreviewmaketitle

	\IEEEraisesectionheading{\section{Introduction}\label{sec:introduction}}

\IEEEPARstart{M}{illions} of drivers provide transportation services
for over ten million passengers every day at Didi
Chuxing~\cite{didi}, which is a Chinese counterpart of
UberPOOL~\cite{uber}.
In peak travel periods, Didi needs to match more than a
hundred thousand passengers to drivers every
second~\cite{zhang2017taxi}, and rider demand often greatly exceeds
rider capacity.
Two approaches can be used to mitigate this problem.
The first method attempts to predict areas with high travel
demands using historical data and statistical predictions or a
heat map, and taxis are strategically deployed in the corresponding
areas in advance.
An alternative approach is to serve multiple riders with
fewer vehicles using a ridesharing service: riders with similar
routes and time schedules can share the same vehicle
\cite{teubner2015economics,alonso2017demand}.
According to statistical data from the Bureau of Infrastructure,
Transport and Regional Economics~\cite{bitre},
there are less than 1.6 persons per vehicle per kilometer in
Australia.
If only 10\% of vehicles had more than one passenger, then it would reduce
annual fuel consumption by 5.4\%~\cite{jacobson2009fuel}.
Therefore, increasing vehicle occupancy rates would provide many
benefits including the reduction of gas house emissions.
Moreover, it has been reported that a crucial imbalance exists in
supply and demand in peak hour scenarios, where the rider demand is
double the rider availability based on historical data statistical
analysis at Didi Chuxing~\cite{didireport}.
Alleviating traffic congestion challenges during peak commuter times
will ultimately require significant government commitment dedicated
to increasing the region’s investment in core transportation
infrastructure~\cite{downs2005still}.
In this paper, we focus on the dynamic ridesharing problem,
specifically during peak hour travel periods.

From an extensive literature study, we make the following
observations to motivate this work -- (1) Existing ridesharing
studies~\cite{zheng2018order,chen2018priceAndtime} do not fully
explore the scenario where the number of riders is much greater than
the number of available drivers, and so the scalability of current
solutions in this setting remains unclear.
(2) Prior
studies~\cite{ta2018efficient,ma2013t,zheng2018order,huang2014large,cheng2017utility,xing2009smize}
primarily focus on single objective optimization solutions, such as
minimizing the total travel distance from the perspective of
drivers~\cite{ma2013t,huang2014large}, or maximizing
the served rate of ridesharing system~\cite{xing2009smize}.
In contrast, we aim to optimize two objectives: maximize the served
rate and minimize the total additional distance.
(3) Previous
studies~\cite{ma2013t,fu2017top,tong2018unified,chen2018priceAndtime,cao2015sharek}
report the processing time mainly based on the rider request matching
time, but not the \emph{trip schedule update time}, which encompasses
a driver's current trip schedule and the underlying index structure
updates.
However, in a dynamic ridesharing scenario where vehicles are
continuously moving, accounting for these additional costs produces
a more realistic comparison of the algorithms being studied.

Our goal in this work is to determine a series of trip schedules with the
minimum total additional distance capable of accommodating as many
rider requests as possible, under a set of spatio-temporal
constraints.
In order to achieve this goal, we must account for three features:
(1) each driver has an initial location and trip schedule, but rider
requests are being continuously generated in a streaming manner; (2)
before a driver receives a new incoming rider request, the arrival
time and location of the rider are unknown; (3) the driver and the
rider should be informed at short notice about whether a matching is
possible or not.
Specifically, the driver should be notified quickly if a new rider is
added, and similarly the rider should be informed quickly if her
travel request can be fulfilled based on the current preference
settings, such as waiting time tolerance to be picked up.

The following challenges arise in addressing this problem: 

(1)\textit{ How can we find eligible driver-rider matching pairs to
maximize the served rate and also minimize the additional
distance? } 
If more rider requests are satisfied, it implies that the driver has
to travel further to pick up more rider(s).
For instance, if a driver $d$ already has passengers on board but
receives a new rider request $r_1$, the driver might detour to pick
up $r_1$, resulting in an increase in the served rate and the
distance traveled.

(2)\textit{ How can we make a decision on the best service sequence
for a trip schedule?}
Every rider has their own maximum allowable waiting time, and detour
time tolerance.
When a rider request is served, these constraints should not be
violated.
For example, a rider request $r_1$ is already in the trip schedule of
a driver $d$, but a new rider request $r_2$ is received while the
driver is serving $r_1$.
The driver needs to determine if $r_2$ can be picked up first without
violating $r_1$'s constraints.

(3)\textit{ How can we efficiently support the ridesharing problem
for streaming rider requests?}
The time used for determining the updated trip schedule as per new
rider requests should not exceed the time window size.

In order to address the above challenges, we make the following
contributions:
\begin{itemize}[leftmargin=*]
\item We define a variant of the dynamic ridesharing problem, which
aims to optimize two objectives subject to a series of
spatio-temporal constraints presented in Section
\ref{sec:def_pf}. 
In addition, our work mainly focuses on a common yet important
scenario where the number of drivers is insufficient to serve all
riders in peak travel periods. 

\item We develop an index structure on top of a partitioned road
network and compute the lower bound of the shortest path distance
between any two vertices in constant time in Section \ref{sec:is}. 
We then propose a pruning-based rider request insertion algorithm
based on several pruning rules to accelerate the matching process in
Section \ref{sec:riderinsertion}.

\item We further propose two algorithms to find matchable and eligible
driver-rider pairs, which aim to maximize the served rate and
minimize the total additional distance, in Section
\ref{subsec:df} and Section \ref{subsec:gr} respectively.

\item We conduct extensive experiments on a real-world
dataset in order to validate the efficiency and effectiveness of our
methods under different parameter settings in Section
\ref{sec:exp}.
\end{itemize}

In addition, we review the related work in Section~\ref{sec:related}
and conclude the work in Section~\ref{sec:conclusion}.
\vspace{-8px}

	\section{Related Work} \label{sec:related}

Ridesharing has been intensively studied in recent years, in both
\textit{static} and \textit{dynamic} settings.
A typical formulation for the \textit{static} ridesharing problem
consists of designing drivers' routes and schedules for a set of rider
requests with departure and arrival locations known beforehand
\cite{ma2013analysis,lin2012research,yan2011optimization,wang2016pickup},
while the \textit{dynamic} ridesharing problem is based on the
setting that new riders are continuously added in a stream-wise
manner~\cite{hall1997dynamic,tao2007dynamic,chan2012ridesharing,morency2007ambivalence}.

\textbf{Static Ridesharing Problem. }
Ta et al.~\cite{ta2018efficient} defined two kinds of ridesharing models to
maximize the shared route ratio, under the assumption that at most
one rider can be assigned to a driver in the vehicle. 
Cheng et al.~\cite{cheng2017utility} proposed a utility-aware
ridesharing problem, which can be regarded as a variant of the
dial-a-ride problem~\cite{cordeau2007dial,psaraftis1980dynamic}.
The aim was to maximize the riders' satisfaction, which is defined as
a linear combination of vehicle-related utility, rider-related
utility, and trajectory-related utility,
such as rider sharing preferences.
Bei and Zhang~\cite{bei2018algorithms} investigated the assignment
problem in ridesharing and designed an approximation algorithm with a
worst-case performance guarantee when only two riders can share a
vehicle.
In contrast, we focus primarily on the \textit{dynamic} ridesharing
problem, and allow the maximum number of riders to be greater than
two.

\textbf{Dynamic Ridesharing Problem. }
Several existing techniques are well illustrated and outlined by two
recent surveys
\cite{agatz2012optimization,furuhata2013ridesharing}. 
We describe the literature in chronological order.
Agatz et al.~\cite{agatz2011dynamic} explored a ride-share
optimization problem in which the ride-share provider aims to
minimize the total system-wide vehicle-miles.
Xing et al.~\cite{xing2009smize} studied a multi-agent system with
the objective of maximizing the number of served riders.
Kleiner et al.~\cite{kleiner2011mechanism} proposed an auction-based
mechanism to facilitate both riders and drivers to bid according to
their preferences. 
Ma et al.~\cite{ma2013t,ma2015real} proposed a dynamic taxi
ridesharing service to serve each request by dispatching a taxi with
the minimum additional travel distance incurred.
Huang et al.~\cite{huang2014large} designed a kinetic tree to
enumerate all possible valid trip schedules for each driver.
Duan et al.~\cite{duan2016real} studied the personalized ridesharing
problem which maximizes a satisfaction value defined as a
linear function of distance and time cost.
Zheng et al.~\cite{zheng2018order} considered the platform profit as
the optimization objective to be maximized by dispatching the orders
to vehicles.
Tong et al.~\cite{tong2018unified} devised route plans to maximize
the unified cost which consists of the total travel distance and a
penalty for unserved requests.
Chen et al.~\cite{chen2018priceAndtime} considered both the pick-up
time and the price to return multiple options for each rider request
in dynamic ridesharing.
Xu et al.~\cite{xu2019efficient} proposed an efficient
rider insertion algorithm that minimizes the maximum flow time of all
requests or minimizes the total travel time of the driver.

Our work is different from existing ridesharing studies in two
aspects:
(1) Our focus is the peak hours scenario where there are too few
drivers to satisfy all of the rider requests, which has not been
explored in previous work.
Later in the experimental study we extend the
state-of-the-art~\cite{chen2018priceAndtime} to the peak hour
scenario and show that our approach outperforms it in both efficiency
and effectiveness.
(2) Most previous studies aim to optimize a single
objective~\cite{huang2014large,agatz2011dynamic,ma2013t,xing2009smize,zheng2018order,cheng2017utility,xu2019efficient} or a customized linear function~\cite{duan2016real,tong2018unified}.
This differs from our problem scenario where we solve the
dual-objective optimization problem.
Chen et al.~\cite{chen2018priceAndtime} also considered two criteria
(price and pick-up time) from the perspective of riders, but not the
same criteria (served rate and additional distance) as we use to
satisfy riders, drivers and ridesharing system requirements.
In our problem, riders can provide their personal sharing preferences
by giving constraint values, such as waiting time tolerance,
drivers and ridesharing system expect to
serve more riders with less detour distance, which coincides with our
goal, i.e., maximizing the served rate
and meanwhile minimizing the additional distance.
Kleiner et al. \cite{kleiner2011mechanism} took both minimizing the total travel
distance and maximizing the number of served riders into
consideration, where riders and drivers can select each other by adjusting bids. 
They assumed that each driver can only share the trip
with one other rider, which limited the potential of ridesharing system. 
Other recent studies on task assignment
\cite{tong2016onlineminimum,tong2016onlinemobile} exploit bilateral
matching between a set of workers and tasks to achieve one single
objective, minimizing the total distance
\cite{tong2016onlineminimum} or maximizing the total utility
\cite{tong2016onlinemobile}.
However, when the ordering of multiple riders must also be mapped
to a single driver, bilateral mapping is not sufficient. 


\vspace{-5px}

\section{Problem formulation}\label{sec:def_pf}
\subsection{Preliminaries}
Let $G = \langle V,E \rangle$ be a road network represented as a
graph, where $V$ is the set of vertices and $E$ is the set of edges.
We assume drivers and riders travel along the road network.
Let $D$ be a set of drivers, where each $d \in D$ is a tuple $\langle
\loc,\ms, tr \rangle$.
Here, \loc $\in V$ is the current location, \ms \xspace is the
maximal seat capacity, and $tr$ is the current trip schedule of the
driver $d$.
If {\loc} does not coincide with a vertex, we map the location to the
closest vertex for ease of computation.
Each trip schedule $tr=\langle o_0,o_1,o_2,\ldots,o_n \rangle$ is a
sequence of points, where $o_0$ is the driver's current location, and
$o_k$ ($1 \leq k \leq n$) is a source or destination point of a
rider. 
We assume that the rider requests arrive in a streaming fashion.
We impose a time-based window model, where we process the set of
requests that arrive in the most recent timeslot.

\begin{definition}
(\textit{Rider Request}).
Let $R$ be a set of rider requests.
Each $r \in R$ is a tuple
$\langle t,\ridersource,\riderdes,rn,w,\detour \rangle$, where $t$ is
the request submission time, $\ridersource \in V$ is the source
location, $\riderdes \in V$ is the destination
location, $rn$ is the number of riders in that request, $w$ is the waiting time threshold (i.e., the maximum period $r$ needs to be picked up after submitting a request), and {\detour} is a detour time threshold
(explained later in Definition~\ref{def:problem}).
\end{definition}
\vspace{-5px}
\myparagraph{Additional distance} 
As mentioned before, each driver $d \in D$ maintains a trip
schedule $tr=\langle o_0,o_1,o_2,\ldots,o_n \rangle$, where the
corresponding riders are served sequentially in $tr$.
If a new rider must be served by $d$, the trip schedule changes. 
The additional distance $\additiondis =
\dist_{\ntripschedule}-\dist_{\tripschedule}$ is defined as the
difference between the travel distance of the updated trip schedule
{\ntripschedule} after a new rider request is inserted, and the
travel distance of the original trip schedule {\tripschedule}.

Here, the travel distance of a trip schedule is computed as:
$dis_{tr}=\sum_{k=1}^{n}\dist(o_{k-1},o_k)$, where the
distance between two points is computed as the shortest path in the
road network $G$.
For any two points $o_i$ and $o_j$ ($0 \leq i \textless j \leq n$)
which are not adjacent in \tripschedule, the travel distance from
$o_i$ to $o_j$ is defined as follows:
$dis_{tr}(o_i,o_j)=\sum_{k=i}^{j-1}\dist(o_k,o_{k+1})$.

\myparagraph{Served Rate}
The served rate $SR = |\sriders| / |R|$ is defined as the ratio of the number of
served riders ${\sriders}$ (i.e., matched with a driver) over the total number of riders $|R|$. 
Now we formally define the dynamic ridersharing problem.
\vspace{-5px}
\subsection{Problem Definition}
\begin{definition}\label{def:problem}
(\textit{Dynamic Ridesharing}).
Given a set of drivers $D$ and a set of new incoming riders $R$ on
road network, the dynamic ridesharing problem finds the
optimal driver-rider pairs such that (1) the served rate
$SR$ is maximal; (2) the total additional distance $\sum_{i=1}^{|D|} AD_{i}$ is minimal, subject to the following
spatio-temporal constraints: 

(a) \textbf{Capacity constraint.}
The number of riders served by any driver $d$ should not exceed the
corresponding maximal seat capacity $c$.

(b) \textbf{Waiting time constraint.}
The actual time for the driver to pick up the rider after receiving
the request should not be greater than the rider's waiting time
constraint $w$.

(c) \textbf{Detour time constraint.}
A driver may detour to pick up other riders, so the actual
travel time {\timesd} by any rider $r$ in the road network
should be bounded by the shortest travel time {\timeminsd} multiplied by
the corresponding rider's detour threshold {\detour}, which is 
$\timesd \leq (1+\detour)\times \timeminsd$. 
\end{definition}

Note that time and distance are interchangeable using
reasonable travel speeds collected from historical data. 
Therefore, we emphasize the following two points for clarity of exposition: (i) we adopt a uniform travel speed assumption henceforth, while we conduct the experiments under different settings of travel speed in experiments (Section \ref{sec:exp}) to simulate varying road conditions in peak hours; (ii) in accordance with our index structure (Section \ref{sec:is}), which is a distance-based framework, $w$ is stored as distance value computed by a multiplication of waiting time and travel speed. In regard to the detour time constraint, we represent it as an inequation based on distance, i.e., $\distsd \leq (1+\detour)\times \distminsd$, where \distsd \xspace is the actual
travel distance, and \distminsd \xspace is the shortest travel distance. 

\vspace{-2ex}
\subsection{Solution Overview}
The dynamic ridesharing problem is a classical constraint-based
combinatorial optimization problem, which was proven to be
NP-hard~\cite{zheng2018order}. 
Our objective is to match the set of new rider requests $R$
in each timeslot with the set of drivers $D$ and update the
corresponding drivers' trip schedules, such that all constraints are met, the served rate is maximized, and the total
additional distance is minimized.

We first propose an underlying index structure on top of a
partitioned road network to compute the lower bound distance between
any two vertices in Section \ref{sec:is}.
Then, one crucial problem is to accommodate a new incoming rider
request $r$ in an existing trip schedule for a driver efficiently.
We present a pruning based algorithm to efficiently insert the source
and destination points of a rider into an existing trip schedule of a
driver in Section \ref{sec:riderinsertion}.
Note that, when inserting new points into a trip schedule, the
constraints of existing riders in that trip schedule cannot be
violated.
Moreover, multiple drivers may be able to serve a rider, and there
can be multiple options to insert a rider's source and destination
points into a trip schedule.
Thus, the dynamic insertion of a rider's request into a trip schedule is a
difficult optimization problem.

Next, we propose two different algorithms in Section~\ref{sec:mb} to
find the match between $R$ and $D$ using the insertion algorithm.
The first is the \algoDistancefirst algorithm in
Section~\ref{subsec:df}.
Specifically, we process each rider request one by one in a
first-come-first-serve manner according to the request submission time.
For each rider request, we invoke the insertion algorithm
(Algo.~\ref{riderInsertion}) to find an eligible driver who
generates the minimal additional distance.
Although {\algoDistancefirst} can match each rider request with a
suitable driver efficiently, the served rate of the ridesharing
system is neglected in the process.
Therefore, we propose the {\algoGreedy} algorithm in
Section~\ref{subsec:gr}.
In this approach, we consider the batch of rider requests within the
most recent timeslot (e.g., 10 seconds) altogether and match them with
a set of drivers optimally by trading-off two metrics: served rate
and additional distance.
\vspace{-10px}

	\vspace{-2ex}
\section{Index Structure}\label{sec:is}
\begin{table}[t]
\renewcommand{\arraystretch}{1.3}
\caption{Symbol and description}
\label{table:notation}
\vspace{-2ex}
\centering
\begin{tabular}{lp{5.6cm}}
\hline
{\bfseries Symbol} & {\bf Description}\\
\hline

$d=\langle id,\loc,\ms,tr \rangle$ & A driver $d$ with a unique $id$, current location $\loc$, maximal seat capacity $\ms$, and current trip schedule $tr$ \\ \hline

$tr=\langle o_0,o_1,\ldots,o_n \rangle$	&	A trip schedule $tr$ consists of a sequence of points, where $o_0$ is the driver's current location, and $o_i$ ($1 \leq i \leq n$) is a source or destination point of a rider \\ \hline

$r=\langle t,\ridersource,\riderdes,rn,w,\detour \rangle$ & A rider request $r$ with the submission time $t$, a source point \ridersource, a destination point \riderdes, the number of riders $rn$, a waiting time threshold $w$, and a detour threshold $\detour$ \\ \hline


$\distactual(s, d)$ & The actual travel distance between $s$ and $d$\\ \hline

$\dist(s, d)$ & The shortest travel distance between $s$ and $d$\\ \hline




$\dist(G_i,G_j)$	&	The lower bound distance between two subgraphs $G_i$ and $G_j$	\\ \hline

$dis^{\downarrow}(u,G_v)$	&	The lower bound distance between a vetex $u$ and a subgraph $G_v$\\ \hline

$dis^{\downarrow}(u,v)$	&	The lower bound distance between any two vertices $u$ and $v$ \\ \hline

$\detourdis(o_a,o_b,o_c)$	&	The incremental distance by inserting $o_b$ between $o_a$ and $o_c$, then $\detourdis(o_a,o_b,o_c)=\dist(o_a,o_b) + \dist(o_b,o_c)-\dist(o_a,o_c)$	\\ \hline

\utilitydf, \utilitygr	&	Two optimization utility functions	\\ \hline

$\updatetimeinter$	&	The update time window	\\ \hline

$\speed$	&	The travel speed	\\ \hline
\end{tabular}
\vspace{-10pt}
\end{table}

In this section, we propose a new index structure on top of a partitioned
road network.
First we present the motivation behind the index design, and then we
present the details of the index in Section~\ref{index-roadnetwork}. Table~\ref{table:notation} presents the notation used throughout this
work.
\vspace{-7px}
\subsection{Motivation}
The distance computation from a driver's location to a rider's pickup
or drop-off point can be reduced to the shortest path computation
between their closest vertices in the road network, which can be easily
solved using an efficient hub-based labelling
algorithm~\cite{abraham2011hub}.
Although invoking the shortest path computation once only requires a
few microseconds, a huge number of online shortest path computations
are required when trying to optimally match new incoming riders with
the drivers with constraints, and update the trip schedules
accordingly.
Such computations lead to a performance bottleneck.

A straightforward way is to precompute the shortest path distance
offline for all vertex pairs and store them in memory or disk.
Then the shortest path query problem is simply reduced to a direct
look-up operation.
Although the query can be processed efficiently, this approach is
rarely used in a large road network in practice, especially when many variables
may change in a dynamic or streaming scenario.
Therefore, it is essential but non-trivial to devise an efficient
index over road network which can be used to estimate the actual
shortest path distance between any two locations.

Since road networks are often combined with non-Euclidean distance
metrics, a traditional spatial index cannot be directly used. 
For example, a grid index is widely used in existing ridesharing
studies~\cite{chen2018priceAndtime,tong2018unified,ma2013t,cao2015sharek}
to partition the space.
Generally, they divide the whole road network into multiple
equal-sized cells and then compute the lower or upper bound
distance between any two grids.
These distance bounds are further used for pruning.
However, as the density distribution of the vertices vary widely in
urban and rural regions, most grids are empty, and contain no 
vertex.
For example, more than 80\% grids are empty in the grid index, resulting in very weak  pruning power in ridesharing scenarios \cite{chen2018priceAndtime}.
Although a quadtree index can divide the road network structure
in a density-aware way, it has to maintain a consistent
hierarchical representation such that each child node update may lead
to a parent node update, which can increase the update costs when
available drivers are moving (which is often true in real world
scenarios).
Therefore, we choose to adopt a density-based road network partitioning
approach, which can efficiently estimate shortest path distances.
\vspace{-10px}
\subsection{Road Network Index}\label{index-roadnetwork}
The index is constructed in two steps --
(1) Partition the road network into subgraphs such that 
closely connected vertices are stored in the same subgraph, while the
connections among subgraphs are minimized.
(2) For each subgraph, it stores the information necessary to
efficiently estimate the shortest path between any pair of
vertices in the road network.

\begin{figure*}[t]
	\centering
	\begin{minipage}{0.62\textwidth}
		\subfloat[Road Network Partition example.
		The red vertices represent the bridge vertices in each
		subgraph, while the red edges denote the cut edges connecting
		two disjoint subgraphs.\label{subfig-graphpartition}]
		{\includegraphics[width=0.46\linewidth]{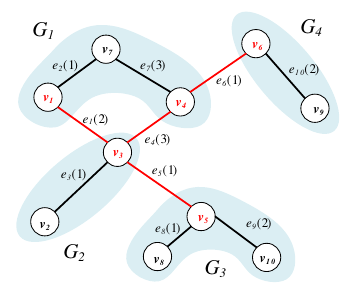}}
		\subfloat[Road Network Index example \label{subfig-graphindex}]{
			\includegraphics[width=0.46\linewidth]{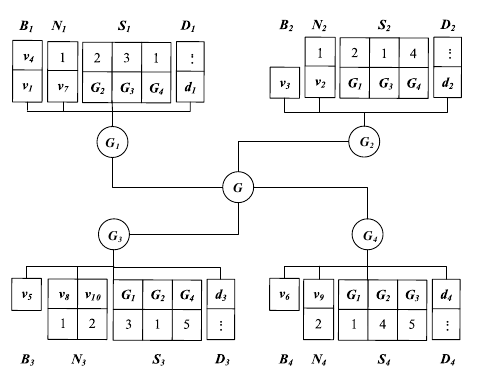}}
		\caption{An illustration example of road network indexing}
		\label{fig-roadnetwork} 
	\end{minipage}
	\begin{minipage}{0.375\textwidth}
		\subfloat[The source insertion position $i$ and the
		destination insertion position $j$ are both after the $n$-th
		position (after finishing the existing trip schedule)
		\label{insertion_1}]{
			\includegraphics[width=0.46\linewidth]{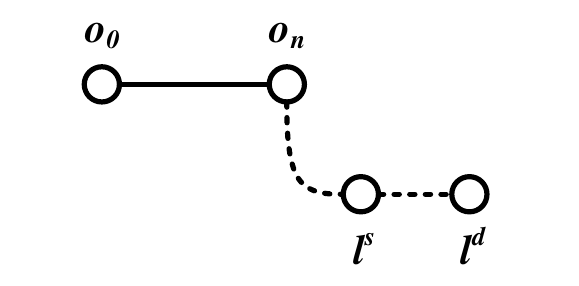}}
		\quad
		\subfloat[$i$ and $j$ are consecutive and inserted before the $n$-th position \label{insertion_2}]{
			\includegraphics[width=0.43\linewidth]{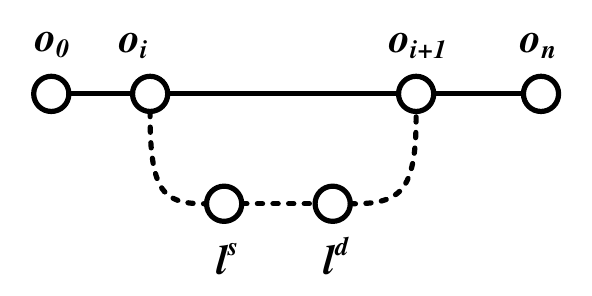}}
		\vspace{-0.08in} 
		\subfloat[$i$ and $j$ are not consecutive and inserted before the $n$-th position\label{insertion_3}]{
			\includegraphics[width=0.46\linewidth]{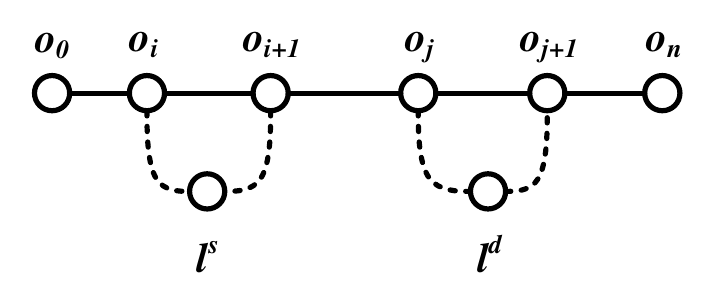}}
		\quad
		\subfloat[$i$ is before the $n$-th position but $j$ is inserted after the $n$-th position\label{insertion_4}]{
			\includegraphics[width=0.43\linewidth]{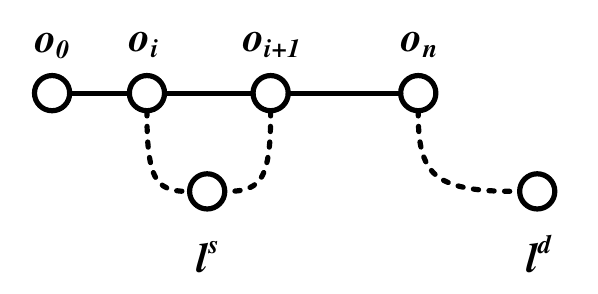}}
		\caption{Four different ways the source and the destination
			points of a new rider can be inserted in an existing trip
			schedule}
		\label{fig-insertion} 
	\end{minipage}
	\vspace{-10px}
\end{figure*}

\myparagraph{Density-based Partitioning}
We present a density-based partitioning approach which divides the
road network $G$ into multiple subgraphs.
The partition process follows two criteria: (1) \textbf{Similar
number of vertices:} each subgraph maintains an approximately equal
number of vertices such that the density distributions of the
subgraphs are similar.
(2) \textbf{Minimal cut:} an edge $e$ is considered as a \textit{cut}
if its two endpoints belong to two disjoint subgraphs.
The number of cuts is minimal, so that the vertices in a subgraph are
closely connected.

Given a road network $G= \langle V,E \rangle$ with a vertex set $V$
and an edge set $E$, a partition number \partitionnum, $G$ is divided
into a set of subgraphs $G_1, G_2, \ldots, G_\tau$, where each
subgraph $G_u$ contains a subset of vertices $V_u$ and a subset of
edges $E_u$ such that (1) $V_1 \cup V_2 \cup \ldots
\cup V_{\tau} = V$, $E_1 \cup E_2 \cup \ldots 
\cup E_{\tau} \subseteq E$; (2) if $1 \leq u, v \leq \tau$, and $u
\neq v$, then $V_u \cap V_v = \emptyset$, $E_u \cap E_v = \emptyset$;
(3) $|E_1|
\approx |E_2| \approx \ldots 
\approx |E_{\tau}|$; (4) the number of edge cuts
$|E|-\sum_{u=1}^{\tau}|E_u|$ is minimized.

The graph partitioning problem has been well-studied in the literature 
and is not the primary focus of this paper.
Instead, we focus on how to define the distance bounds in order to
reduce unnecessary shortest path distance computations.
Thus we use a state-of-the-art method~\cite{karypis1998fast} to
obtain the density based partitioning of the road network $G$.


\myparagraph{Subgraph information}
After we obtain a set of disjoint subgraphs from $G$, we build our
index by storing the following information for each subgraph $G_i$.

\begin{enumerate}[leftmargin=*]
\item \textbf{Bridge Vertex Set.}
If an edge $e$ is a \textit{cut}, then each of its two endpoints
is regarded as a \textit{bridge vertex}.
We store the set $B_i$ of the bridge vertices of each subgraph $G_i$.

\item \textbf{Non-Bridge Vertex Set.}
Vertices not in $B_i$ of a subgraph $G_i$ are
stored as a {\em non-bridge vertex} set $N_i$, i.e., $V_i \backslash
B_i=N_i$.
For each non-bridge vertex $u \in N_i$, we use $dis^{\downarrow}(u)$
to denote the lower bound distance of $u$, which is the shortest path
distance to its nearest bridge vertex $v \in B_i$ in the same
subgraph $G_i$.

\item \textbf{Subgraph set.} 
For each subgraph $G_i$, we store a list $S_i$ of the lower bound
distances, one entry for each of the other subgraphs.
Specifically, as the vertex of a subgraph is reachable from a vertex
of another subgraph only through the bridge vertices, the lower bound
distance \dist$(G_i,G_j)$ is calculated with the following equation.
\begin{equation}
\label{equ-lowdisTwographs}
\text{\dist}(G_i,G_j)=\left\{\begin{matrix}
0  & \text{if }G_i=G_j \\ 
\min_{u,v} \text{\dist}(u,v), \\
u \in B_i, v \in B_j  & \text{otherwise} \\
\end{matrix}\right.
\end{equation}

\item \textbf{Dispatched Driver Set.}
If the current location of a driver $d$ is situated in a subgraph
$G_i$, then $d$ is stored in the dispatched driver set $D_i$.
\end{enumerate}
\vspace{-14px}
\subsection{Bounding Distance Estimations}

Furthermore, we introduce two additional concepts on how to calculate
the lower bound distance from a vertex to a subgraph (Definition
\ref{def-lowdisTograph}) or another vertex (Definition
\ref{def-lowdisTovertex}).
\begin{definition}
	\label{def-lowdisTograph}
	(\textit{Lower Bound Distance Between a Vertex and a Subgraph}). 
  Given a vertex $u$ which belongs to $G_u$ and a subgraph $G_v$, the
  lower bound distance $dist^{\downarrow}(u,G_v)$ from $u$ to $G_v$
  is defined as follows:
  \begin{equation}
	\label{lowerdistance_sc}
	dis^{\downarrow}(u,G_v)=\left\{\begin{matrix}
	0 & \text{if }G_u=G_v \\ 
	\text{\dist}(G_u,G_v) & \text{if }G_u \neq G_v, u \in B  \\
	\text{\dist}(G_u,G_v)+dis^{\downarrow}(u) & \text{if }G_u \neq G_v, u \notin B  \\
	\end{matrix}\right.
	\end{equation}
\end{definition}


\begin{definition}
	\label{def-lowdisTovertex}
	(\textit{Lower Bound Distance Between Two Vertices}). 
	Given a vertex $u$ which belongs to $G_u$ and a vertex $v$ which is
	located at $G_v$, the lower bound distance $dist^{\downarrow}(u,v)$
	is defined as follows:
	\vspace{-2ex}
	\begin{equation}
		\small 
		\label{lowerdistance_sd}
		dis^{\downarrow}(u,v)=\left\{\begin{matrix}
		0 & \text{if }G_u=G_v \\ 
		\text{\dist}(G_u,G_v) & \text{if }G_u \neq G_v, u \in B_u, v \in B_v  \\
		\text{\dist}(G_u,G_v)+dis^{\downarrow}(u) & \text{if }G_u \neq G_v, u \notin B_u, v \in B_v  \\
		\text{\dist}(G_u,G_v)+dis^{\downarrow}(v) & \text{if }G_u \neq G_v, u \in B_u, v \notin B_v  \\
		\text{\dist}(G_u,G_v)+dis^{\downarrow}(u) \\ +dis^{\downarrow}(v) & \text{otherwise} \\
		\end{matrix}\right.
	\end{equation}
\end{definition}

In our implementation, we construct the road network index offline
and compute $dis^{\downarrow}(u,G_v)$ or $dis^{\downarrow}(u,v)$
online.
The index structure is memory resident, which includes
\dist$(G_u,G_v)$ and $dis^{\downarrow}(u)$ information.
Therefore, $dis^{\downarrow}(u,G_v)$ (or $dis^{\downarrow}(u,v)$) can
be computed using Eq. \ref{lowerdistance_sc} (or Eq.
\ref{lowerdistance_sd}) in $O(1)$ time complexity.
\begin{example}
A road network partitioning example is depicted in
Fig.~\ref{subfig-graphpartition}.
There are ten vertices and ten edges in the graph, and they are
divided into four subgraphs: $G_1$, $G_2$, $G_3$, and $G_4$.
The values in parenthesis after each edge denote the distance.
The red vertices represent the bridge vertices in each subgraph,
while the red edges denote the cut edges connecting two disjoint
subgraphs.
E.g., $e_1$ is a cut because it connects two
disjoint subgraphs $G_1$ and $G_2$.
$v_1$ and $v_3$ are two bridge vertices because they are two
endpoints of a cut $e_1$.

The corresponding road network index structure is illustrated in
Fig.~\ref{subfig-graphindex}.
For a subgraph $G_1$, $B_1$ includes two bridge vertices: $v_1$ and
$v_4$.
There is one remaining non-bridge vertex $v_7$ which belongs to
$N_1$.
Then the lower bound distance $dis^{\downarrow}(v_7)$ from $v_7$ to a
bridge vertex in the same subgraph is
$\min\{$\dist$(v_7,v_1),$\dist$(v_7,v_4)\}=1$.
Three subgraphs are connected with $G_1$ directly or indirectly
through cut edges.
The distance between $G_1$ and $G_2$ is
$\dist(G_1,G_2)=\min\{\dist(v_1,v_3),\dist(v_3,v_4)\}=2$,
and the distance between $G_1$ and $G_4$ is
$\dist(G_1,G_4)=\min\{\dist(v_4,v_6)\}=1$.

The lower bound distance $dis^{\downarrow}(v_7,G_3)$ between
$v_7$ and $G_3$ is
$\dist(G_1,G_3)+dis^{\downarrow}(v_7)=3+1=4$.
The lower bound distance $dis^{\downarrow}(v_7,v_9)$ between $v_7$
and $v_9$ is
$\dist(G_1,G_4)+dis^{\downarrow}(v_7)+dis^{\downarrow}(v_9)=1+1+2=4$.
\end{example}

	\vspace{-2ex}
\section{Pruning-based Rider Request Insertion}\label{sec:riderinsertion}

In this section, we propose a rider insertion algorithm on top of
several pruning rules to insert a new incoming rider request $r$ into
an existing trip schedule of a driver $d$ efficiently, such that a
customized utility function \ensuremath{U} is minimized, i.e.,
{\utilitydf} (Eq.~\ref{df:utilitydf}) for {\algoDistancefirst}
algorithm and {\utilitygr} (Eq. \ref{gr:utility}) for
{\algoGreedy} algorithm, respectively.
\vspace{-6.5px}
\subsection{Problem Assumption}\label{sec:ri:as}
First, same as prior work~\cite{cheng2017utility, zheng2018order,
ma2015real}, when receiving a new rider request, a driver maintains
the original, unchanged trip schedule sequence.
In other words, we do not reorder the current trip schedule to
ensure a consistent user experience for riders already scheduled.
For example, if a driver has been assigned to pick up $r_1$ first and
then pick up $r_2$, then the pickup timestamp of $r_1$ should be no
later than the pickup timestamp of $r_2$.
Second, in contrast to restrictions commonly adopted in prior
work~\cite{zheng2018order,
bei2018algorithms,kleiner2011mechanism,duan2016real},
where the number of
rider requests served by each driver is never greater than two, we assume
that more than two riders can be served as long as the seat capacity
constraint is not violated, which improves the usability and scalability
of the ridesharing system.
\vspace{-2.5ex}
\subsection{Approach Description}\label{sec:ri:des}
Given a set of drivers $D$ and a new rider $r$, we aim to find a
matching driver and insert {\source} and {\des} of $r$ into a 
driver's trip schedule.
A straightforward approach can be applied as follows: (1) for each
driver candidate $d$, we enumerate all possible insertion
positions for {\source} and {\des} in $d$'s current trip schedule
(its complexity is $O(n^2)$, where $n$ is the number of points in the
trip schedule); (2) for each possible insertion position pair
$\langle i,j \rangle$, we check whether it violates the waiting and
detour time constraints of both $r$ and the other riders who have been
scheduled for $d$ ($O(n)$).
Therefore, the total time consumption is $O(n^3)$.
As the insertion process is crucial to the overall efficiency, we
now propose several new pruning strategies.
Then we present our algorithm for rider insertion derived from these
pruning rules.

Before presenting the pruning strategies, we would like to introduce how the rider request is inserted and a preliminary called the slack distance.

Suppose that {\source} and {\des} are inserted in the $i$-th and
$j$-th location respectively, where $i \leq j$ must hold.
There are four ways to insert them as shown in
Fig.~\ref{fig-insertion}.
To accelerate constraint violation checking, we borrow the idea
of ``slack time'' \cite{huang2014large, tong2018unified} and define
``slack distance''.
\begin{definition}
	\label{def-slackdis}
	(\textit{Slack Distance}). 
  Given a trip schedule $tr=\langle o_0,o_1,o_2,\ldots,o_n \rangle$
  where $o_0$ is the driver's current location, the slack distance
  {\slack} (Eq. \ref{equ:sl}) is defined as the minimal surplus
  distance to serve new riders inserted before the $k$-th position in
  $tr$.
  The surplus distance w.r.t.\ a particular point $o_x$ ($k
  \leq x \leq n$) after the $k$-th position is discussed as follows:

  (1) If $o_x$ is a source point ($l_x^s$) of a rider request $r_x$,
  our only concern is whether the waiting time constraint $w_x$ will be
  violated.
  The actual pickup distance of $r_x$ from the driver's current
  location is $\distactual(o_0,o_x)$, following the trip schedule.
  In the worst case, $r_x$ is picked up just within $w_x$.
  Then the surplus distance generated by $o_x$ is
  $w_x-\distactual(o_0,o_x)$.

  (2) If $o_x$ is a destination point ($l_x^d$) of a rider request
  $r_x$, the detour time constraint will be examined.
  Similarly, the actual drop-off distance is
  $\distactual(l_x^s,l_x^d)$ following the trip schedule from the
  source point $l_x^s$.
  The worst case is that $r_x$ is dropped off within the border of
  detour time constraint, which is $(1+\theta_x)\dist(l_x^s,l_x^d)$
  from $l_x^s$.
  Then the surplus distance generated by $o_x$ is
  $(1+\theta_x)\dist(l_x^s,l_x^d) - \distactual(l_x^s,l_x^d)$.
	
	Thus we define the slack distance {\slack} as Eq.~\ref{equ:sl}.
	\begin{equation}
	\label{equ:sl}
	\text{\slack}=\left\{\begin{matrix}
	\min \{sd[k+1], w_k - \text{\distactual}(o_0,o_k)\}  & \text{if $o_k$ is a source}\\
	\min \{sd[k+1], (1+\theta_k)\text{\dist}(l_k^s,l_k^d)\\- \text{\distactual}(l_k^s,l_k^d)\} & \text{otherwise}
	\end{matrix}\right.
	\end{equation}
\end{definition}

The available seat capacity {\capacity} (Eq. \ref{equ:cp}) in
the process of pick-up and drop-off along the trip schedule changes
dynamically.
\begin{equation}
\label{equ:cp}
\text{\capacity}=\left\{\begin{matrix}
cp[k-1]-rn & \text{if it is a source point}\\
cp[k-1]+rn & \text{otherwise}
\end{matrix}\right.
\end{equation}
In addition, we use an auxiliary variable
$\detourdis(o_a,o_b,o_c)$ to indicate the incremental distance by
inserting $o_b$ between $o_a$ and $o_c$, then
$\detourdis(o_a,o_b,o_c)=\dist(o_a,o_b)+\dist(o_b,o_c)-\dist(o_a,o_c)$.
\begin{example}
Given a driver $d=\langle id,\loc,4,tr \rangle$ already carrying a rider
$\riderone=\langle 0,l_1^s,l_1^d,1,3,0.4 \rangle$, and the current
trip schedule $tr=\langle l, l_1^d \rangle$.
On the driver's way to drop off {\riderone}, $d$ receives a new rider
request $\ridertwo=\langle 1,l_2^s,l_2^d,2,5,0.6 \rangle$, then there
are three possible updated trip schedules $tr^{\prime}_1=\langle l,
l_2^s, l_2^d, l_1^d \rangle$, $tr^{\prime}_2=\langle l, l_1^d, l_2^s,
l_2^d \rangle$, and $tr^{\prime}_3=\langle l, l_2^s, l_1^d, l_2^d
\rangle$.
For $tr^{\prime}_1$, we need to check whether the detour time constraint
of {\riderone} and the waiting time constraint of {\ridertwo} are
violated.
For $tr^{\prime}_2$, we need to check whether the waiting time constraint
of {\ridertwo} is violated.
For $tr^{\prime}_3$, we need to check whether the detour time constraint
of {\riderone}, the waiting time constraint and detour time constraint of
{\ridertwo} are violated.
\end{example}
\vspace{-1.5ex}
\subsection{Pruning Rules}\label{index-prunning}
\subsubsection{Pruning driver candidates}
Based on our proposed index in Section~\ref{sec:is}, we can estimate
the lower bound distance from a driver to a rider's pickup point.
According to the waiting time constraint, drivers outside this range can
be filtered out.
\vspace{-0.5ex}
\begin{lemma}
	\label{lemma-one}
	Given a new rider $r$ and a subgraph $G_v$, if
	$dis^{\downarrow}(\source,G_v)>w$, then the drivers who
	are located in $G_v$ can be safely pruned.
	\end{lemma}
\begin{proof}
\vspace{-1.5ex}
The lower bound distance $dis^{\downarrow}(\source,v)$ between the
rider's source point \source \xspace and any vertex $v$ in $G_v$ is
always greater than or equal to the lower bound distance
$dis^{\downarrow}(\source,G_v)$ from \source \xspace to $G_v$.
Therefore, if $dis^{\downarrow}(\source,G_v)>w$ holds, then $\forall
v \in G_v$, $dis^{\downarrow}(\source,v)>w$ also holds.
Thus, the drivers located at any vertex in $G_v$ cannot satisfy the
waiting time constraint of the rider.
\end{proof}
\vspace{-2.5ex}
\begin{lemma}
	\label{lemma-two}
	Given a new rider $r$ and a driver $d$, if
	$dis^{\downarrow}(\source,l)>w$, then $d$ can be safely
	pruned.
	\end{lemma}
\begin{proof}
\vspace{-2ex}
Since the shortest path distance $\dist(\source,l)$ between the new
rider's source point and the driver's current location is no less
than the lower bound distance $dis^{\downarrow}(\source,l)$, we
have $\dist(\source,l)>w$.
\end{proof}

\subsubsection{Pruning rider insertion positions}
It is worth noting that the insertion of a new rider may violate the
waiting or detour time constraint of riders already scheduled for a
vehicle, thus we propose the following rules to reduce the insertion
time examination.
\vspace{-1ex}
\begin{lemma}
	\label{lemma-three}
	Given a new rider $r$, a trip schedule $tr$ and a source
	point insertion position $i$, if $\dist(l,o_i)>w$, then the
	rider should be picked up before the $i$-th point.
\end{lemma}
\vspace{-2ex}
\begin{proof}
Since each vehicle travels following a trip schedule, we have
$dis(l,\source)=dis(l,o_i)+dis(o_i,\source)$, which is not less than
$dis(l,o_i)$.
Then we can get $dis(l,\source)>w$, which violates the waiting time 
constraint of $r$.
\end{proof}

\begin{lemma}
	\label{lemma-four}
	Given a new rider $r$, a trip schedule $tr$ and a source
	point insertion position $i$ ($i<n$), if
	$dis^{\downarrow}(o_i,\source)+dis^{\downarrow}(\source,o_{i+1})-\dist(o_i,o_{i+1})>sd[i+1]$,
	then the rider cannot be picked up at the $i$-th point.
	\end{lemma}
\begin{proof}
	The incremental distance generated by picking up
	{\source} is
	$\detourdis(o_i,\source,o_{i+1})$, which is no
	smaller than
	$dis^{\downarrow}(o_i,\source)+dis^{\downarrow}(\source,o_{i+1})-\dist(o_i,o_{i+1})$.
	Thus, we obtain
	$\detourdis(o_i,\source,o_{i+1})>sd[i+1]$, which implies
	that the incremental distance exceeds the maximal waiting or
	detour tolerance range of point(s) after the $i$-th point.
\end{proof}
\begin{lemma}
	\label{lemma-five}
	Given a new rider $r$ whose source point {\source} is already
	inserted at the $i$-th position of a trip schedule $tr$, and
	a destination point insertion position $j$, if $i=j<n$ (i.e.,
	the insertion position as shown in Fig.~\ref{insertion_2})
	and
	$dis^{\downarrow}(o_i,\source)+\dist(\source,\des)+dis^{\downarrow}(\des,o_{i+1})-\dist(o_i,o_{i+1})>sd[j+1]$,
	then {\des} cannot be inserted at the $j$-th point.
	\end{lemma}
\begin{proof}
	Similar to Lemma \ref{lemma-four}, the additional distance as
	a result of inserting {\source} and {\des} is
	$\dist(o_i,\source)+\dist(\source,\des)+\dist(\des,o_{i+1})-\dist(o_i,o_{i+1})
	$, which is not less than
	$dis^{\downarrow}(o_i,\source)+\dist(\source,\des)+dis^{\downarrow}(\des,o_{i+1})-\dist(o_i,o_{i+1})$.
	Then the additional distance is greater than $sd[j+1]$, which
	violates the waiting or detour time constraint of the point(s)
	after the $j$-th point.
	\end{proof}
\begin{lemma}
	\label{lemma-six}
	Given a new rider $r$ whose source point {\source} is already
	inserted at the $i$-th position of a trip schedule $tr$, and
	a destination point insertion position $j$, if $i<j<n$ (i.e.,
	the insertion position as shown in Fig.~\ref{insertion_3})
	and
	$\detourdis(o_i,\source,o_{i+1})+dis^{\downarrow}(o_j,\des)+dis^{\downarrow}(\des,o_{j+1})-\dist(o_j,o_{j+1})>
	sd[j+1]$, then {\des} cannot be inserted at the $j$-th point.
	\end{lemma}
\begin{proof}
	The incremental distance generated by dropping off $r$ at {\des} is $\detourdis(o_j,\des,o_{j+1})$, 
	which is no smaller than
	$dis^{\downarrow}(o_j,\des)+dis^{\downarrow}(\des,o_{j+1})-\dist(o_j,o_{j+1})$. Then $\detourdis(o_i,\source,o_{i+1})+\detourdis(o_j,\des,o_{j+1})>sd[j+1]$,
	which violates the waiting or detour time constraint of point(s)
	after the $j$-th position.
	\end{proof}
\begin{lemma}
	\label{lemma-seven}
	Given a new rider $r$, a trip schedule $tr$, and two insertion
	positions $i$ and $j$, if $i<j$ and
	$dis^{\downarrow}(\source, o_{i+1})+\dist(o_{i+1},o_j)+dis^{\downarrow}(o_j,\des)>(1+\detour)\dist(\source,\des)$,
	then such two insertion positions for {\source} and {\des} can
	be pruned.
	\end{lemma}
\begin{proof}
	 The actual travel distance $\distactual(\source,\des)$ from
	 {\source} to {\des} is
	 $\dist(\source,o_{i+1})+
	 \dist(o_{i+1},o_j)+\dist(o_j,\des)$, which is not less than $dis^{\downarrow}(\source,o_{i+1})+\dist(o_{i+1},o_j)+dis^{\downarrow}$($o_j,\des)$. Then we can obtain that $\distactual(\source,\des)>(1+\detour)\dist(\source,\des)$, which violates the detour time constraint of $r$.
	 \end{proof}
\begin{lemma}
	\label{lemma-eight}
	Given a new rider $r$, a trip schedule $tr$, and two insertion
	positions $i$ and $j$, $\forall k$ ($i \leq k \leq j$), if
	$rn>\capacity$, then such two insertion positions for {\source}
	and {\des} can be pruned.
\end{lemma}
	\begin{proof}
		It can be easily proved by the capacity constraint. 
	\end{proof}

In practice, we use the lower-bound distance to execute
the pruning rules.
If a possible insertion cannot be pruned, we use the true distance as a
further check, which can guarantee the correctness of the pruning
operation.
\setlength{\algomargin}{0.5em} 
\setlength{\textfloatsep}{0.1cm}  
	\begin{algorithm}[t]  
		\footnotesize 
		\caption{RiderInsertion ($d$, $r$, $U$)}   
		\label{riderInsertion}   
		\begin{algorithmic}[1] 
			\REQUIRE  
			a driver $d$ with a trip schedule \tripschedule $\langle o_0, o_1, \dots, o_n\rangle$, a new rider request $r$, a utility function $U$
			\ENSURE  
			return the utility value if $r$ can be served by $d$; Otherwise, return $-1$.
			\STATE \augdistance$\leftarrow -1$, \optaugdistance$\leftarrow\infty$  \label{ri:initial}	
			\FOR {$i \leftarrow 0$ \TO $n$} 								 \label{ri:for:start}
			\IF[Lemma \ref{lemma-three} and \ref{lemma-eight}] {$cp[i]\geq rn$ or \dist($l$,$o_i$)$\leq w$} 			\label{ri:for:if}
			\IF[Lemma \ref{lemma-four}] {$i \geq n$ and $dis^{\downarrow}$($o_i$,\source)+$dis^{\downarrow}$(\source,$o_{i+1}$)-\dist($o_i$,$o_{i+1}$) $\geq sd[i+1]$} 			\label{ri:for:if2}
			\FOR {$j \leftarrow i$ \TO $n$}	   \label{ri:for:for:start}	
			\IF[Lemma \ref{lemma-eight}] {$cp[j]\geq rn$}   \label{ri:for:for:if}	
			\IF{$j \textless n$}	
			\IF{$i$ = $j$}
			\IF[Lemma \ref{lemma-five}]{$dis^{\downarrow}$($o_i$,\source)+\dist(\source,\des)+$dis^{\downarrow}$(\des,$o_{i+1}$) $-$\dist($o_i$, $o_{i+1}$)$\textgreater sd[j+1]$}  \label{ri:for:for:if:if}	
			\CONTINUE
			\ENDIF
			\ELSE
			\IF[Lemma \ref{lemma-six} and					 \ref{lemma-seven}]{\detourdis($o_i$,\source,$o_{i+1}$)+$dis^{\downarrow}$($o_j$,\des)+$dis^{\downarrow}$(\des, $o_{j+1}$)-\dist( $o_j$,$o_{j+1}$)$\textgreater sd[j+1]$ or $dis^{\downarrow}$(\source,$o_{i+1}$)+\dist($o_{i+1}$,$o_j$)+$dis^{\downarrow}$($o_j$,\des)$\textgreater$ ($1$+\detour)\dist(\source,\des)}				\label{ri:for:for:if:else}	
			\CONTINUE
			\ENDIF
			\ENDIF
			\ENDIF			
			\STATE \augdistance $\leftarrow$ compute the utility value\label{ri:for:for:compute} using $U$ \label{ri:for:for:if:state}
			\IF { \augdistance $\textless$ \optaugdistance }	\label{ri:for:for:compare}
			\STATE \optaugdistance $\leftarrow$ \augdistance \label{ri:for:for:updatedis}
			\ENDIF	
			\ENDIF
			\ENDFOR	
			\ENDIF
			\ENDIF
			\ENDFOR
			\RETURN \optaugdistance
		\end{algorithmic}  
	\end{algorithm}

\vspace{-1em}
\subsection{Algorithm Sketch}\label{sec:ri:al}

The pseudocode for the rider insertion algorithm is shown in
Algo.~\ref{riderInsertion}.
We initialize two local variables: {\augdistance} to record the
utility value found each time, and the best utility value
{\optaugdistance} found so far (line \ref{ri:initial}), where a lower
value denotes better utility.
The pruning rules are first executed for the pickup point {\source}.
We examine whether the vacant vehicle capacity is sufficient to hold
the riders, and that the driver is close enough to provide the
ridesharing service (line \ref{ri:for:if}).
If the conditions hold for the $i$-th position, then we check whether
the detour distance will exceed the slack distance of the following
points starting from the $i$-th position.

If all of the conditions are satisfied, we continue to check the
destination point insertion (line \ref{ri:for:for:start}).
Otherwise, the source point is not added at the $i$-th position.
Similarly, we check whether the current capacity is sufficient (line
\ref{ri:for:for:if}).
Since the sequence between $i$ and $j$ leads to a different detour
distance, two cases are possible: $i=j$ (i.e., the source and the
destination are added to consecutive positions, shown in
Fig.~\ref{insertion_1} and Fig.~\ref{insertion_2}) and $i<j$ (i.e.,
the positions are not consecutive, as shown in Fig.~\ref{insertion_3}
and Fig.~\ref{insertion_4}).
If $i=j$, we compute the incremental distance generated by {\source}
and {\des}, then judge whether it is greater than the slack distance
of the following points after the $i$-th position
(line~\ref{ri:for:for:if:if}).
If $i<j$, besides checking whether the incremental distance generated
by {\source} and {\des} is larger than the slack distance, we also
need to check whether the detour time constraint of $r$ is violated (line
\ref{ri:for:for:if:else}).
If the current utility value obtained is better than the current
optimal value {\optaugdistance} found, we update {\optaugdistance}
(line \ref{ri:for:for:updatedis}).

\textbf{Time complexity analysis.}
The two nested loops (line~\ref{ri:for:start} and
\ref{ri:for:for:start}), each iterating over $n$ points takes
$O(n^2)$ total time.
The examination for constraint violations (line \ref{ri:for:if},
\ref{ri:for:if2}, \ref{ri:for:for:if}, \ref{ri:for:for:if:if}, and
\ref{ri:for:for:if:else}) only takes $O(1)$. 
Calculating the utility function value
\augdistance~(line \ref{ri:for:for:if:state}) also takes $O(1)$,
which will be explained later according to different utility
functions $U$.  
In our method, we use a fast hub-based labeling
algorithm~\cite{abraham2011hub} to answer the shortest path distance
query.
First, a hub label index is constructed in advance by assigning each
vertex $v$ (considered as a ``hub indicator'') a hub label which is a
collection of pairs ($u$, \dist($v$, $u$)) where $u \in V$.
Then given a shortest path query ($s$, $t$), the algorithm considers
all common vertices in the hub labels of $s$ and $t$.
Therefore, the shortest path distance from $s$ to $t$ depends on the
sizes of the hub label sets.
We assume that the time complexity is $O(\Omega)$, where $\Omega$
indicates the average label size of all label sets
\cite{li2017experimental}.
Therefore, the total time complexity is $O(n^2 \shortestpathquery)$.
\vspace{-8px}

	\section{Matching-based Algorithm}\label{sec:mb}
In this section, we introduce two heuristic algorithms
to bilaterally match a set of drivers and rider requests.
Furthermore, we describe the insertion/update process for
dynamic ridesharing.

\vspace{-8px}
\subsection{\algoDistancefirst Algorithm}\label{subsec:df}

The {\algoDistancefirst} algorithm is executed in the following way:
(1) we choose an unserved rider request $r$ with the earliest
submission time from $R$, and perform the pruning-based insertion
algorithm (Algo.~\ref{riderInsertion}) to find a driver $d$ with
minimal additional distance.
(2) If a driver $d$ can be found to match $r$, then the pair $(r,d)$
is joined and added to the final matching result.

The {\algoDistancefirst} algorithm has two important aspects.
First, a matching for a single rider insertion is made whenever the
constraints are not violated, even if the effectiveness is relatively
low.
Second, if multiple drivers are available to provide a ridesharing
service, the {\algoDistancefirst} algorithm always chooses the one
with the minimal additional distance.
Thus, we define the optimization utility function {\utilitydf} as the
additional distance $AD$.
\begin{equation}\label{df:utilitydf}
\text{\utilitydf}=AD
\end{equation}

Based on the insertion option used (Fig.~\ref{fig-insertion}), we
calculate $AD$ in Eq.~\ref{augdistance}.
As the distance $\dist(o_0,o_k)$ from a driver's location to any
existing point in $tr$ is already calculated and stored, the
computation of $\additiondis$ only takes $O(1)$ if the trip schedule
$tr$ and two insertion positions are given.
Specifically, $\additiondis$ is calculated as:
\begin{equation}
	\label{augdistance}
	AD=\left\{\begin{matrix}
	\text{\dist}(o_n,\text{\source})+\text{\dist}(\text{\source},\text{\des})  & \text{if }i=j=n \\ 
	\text{\dist}(o_i,\text{\source})+\text{\dist}(\text{\source},\text{\des})+\text{\dist}(\text{\des},o_{i+1})\\-\text{\dist}(o_i,o_{i+1})  & \text{if $i=j<n$} \\
	\text{\detourdis}(o_i,\text{\source},o_{i+1})+\text{\detourdis}(o_j,\text{\des},o_{j+1})  & \text{if $i<j<n$} \\
	\text{\detourdis}(o_i,\text{\source},o_{i+1})+\text{\dist}(o_n,\text{\des})  & \text{otherwise} \\
	\end{matrix}\right.
\end{equation}

%
%
%

The pseudocode of the {\algoDistancefirst} method is illustrated
in Algo.~\ref{bm:main}, which has two main phases: {\filter} and
{\refine}. We initialize an empty set $L$ to store the current valid
driver-rider pairs (line \ref{bm:initialize:result}).

In the {\filter} phase (lines \ref{bm:for:start}-\ref{bm:for:end}),
for each rider we prune out ineligible driver candidates who violate
the waiting time constraint.
Specifically, for each rider {\rider}, we maintain a driver candidate
list $D_j$ to store the drivers who can possibly serve {\rider}.
We first prune the subgraphs from which it is impossible for a driver
to serve {\rider} (line~\ref{bm:for:if}).
Then we further prune the drivers using a tighter bound, which is the
lower bound distance between the driver's location and the rider's
pickup point (line \ref{bm:for:if:for:if}).
The remaining drivers are inserted into $D_j$ as candidates.

In the {\refine} phase, the retained driver candidates are considered
in the driver-rider matching process.
We select one unserved rider with the earliest submission time
(line~\ref{bm:while:choose}) and find a suitable driver with minimal
additional distance.
In the for loop
(lines~\ref{bm:while:for:start}-\ref{bm:while:for:end}), for each
driver in $D_j$, we examine the feasibility by using
Algo.~\ref{riderInsertion} (line~\ref{bm:while:for:state}) and choose
the one with the smallest additional distance (line
\ref{bm:while:for:if:update}).
If a driver that can satisfy all the requirements is found,
then \rider \xspace is added to the corresponding trip schedule
(line~\ref{bm:while:if:insert}).

\textbf{Time complexity analysis.}
In the {\filter} phase, the two nested iterations
(line~\ref{bm:for:start} and \ref{bm:for:if:for:start}) takes
$O(\tau|R||D_u|)$. 
The time complexity of the {\refine} phase is
$O(|R||D_u|n^2\shortestpathquery)$.
Therefore, the total time complexity is
$O(\tau|R||D_u|+|R||D_u|n^2\shortestpathquery)$.

\setlength{\algomargin}{0.5em} 
\setlength{\textfloatsep}{0.1cm}
\begin{algorithm}[t]   
	\footnotesize
	\caption{\small The Distance-first Algorithm }   
	\label{bm:main}   
	\begin{algorithmic}[1] 
		\REQUIRE  
		a driver set $D$, a rider request set $R$
		\ENSURE  
		an assigned driver-rider pair list $L$;  										
		\STATE $L \leftarrow \emptyset$;													\label{bm:initialize:result}
		\FOR {$G_u \in G$ }																	\label{bm:for:start}
		\FOR {\rider $ \in R$ }																\label{bm:for1}
		\IF[Lemma \ref{lemma-one}] {$dis^{\downarrow}($\rider.\source$,G_u)\leq$\rider.$w$}	\label{bm:for:if}
		\FOR {\driver $ \in G_u.D_u$ }														\label{bm:for:if:for:start}
		\IF[Lemma \ref{lemma-two}] {$dis^{\downarrow}$(\rider.\source,\driver.$l$)$\leq$\rider.$w$}\label{bm:for:if:for:if}
		\STATE $D_j.push($\driver)															\label{bm:for:push}	
		\ENDIF																				\label{bm:for:end}
		\ENDFOR
		\ENDIF
		\ENDFOR
		\ENDFOR  													
		\WHILE {$R \neq \emptyset$} 														\label{bm:while:start}
		\STATE choose one rider \rider \xspace with the earliest submission time			\label{bm:while:choose}
		\STATE $R.pop($\rider$)$ 															\label{bm:while:pop}
		\STATE $d \leftarrow NIL$, $AD_{i,j} \leftarrow -1$, $AD_{i,j}^* \leftarrow\infty$
		\FOR {\driver \xspace in $D_j$}														\label{bm:while:for:start}	
		\STATE $AD_{i,j} \leftarrow$ RiderInsertion(\driver, \rider, \utilitydf)  \label{bm:while:for:state}	
		\IF {$AD_{i,j} \textless AD_{i,j}^*$}					       	 \label{bm:while:for:if}
		\STATE $AD_{i,j}^* \leftarrow AD_{i,j}$, $d \leftarrow $ \driver			\label{bm:while:for:if:update}
		\ENDIF
		\ENDFOR																				\label{bm:while:for:end}	
		\IF {$d \neq NIL$}
		\STATE insert rider \rider \xspace to the current trip schedule of $d$  			\label{bm:while:if:insert}	
		\STATE $L.push(d, $\rider$)$														\label{bm:while:if:push}
		\ENDIF	
		\ENDWHILE																			\label{bm:while:end}
		\RETURN $L$ 				  														\label{bm:end}
	\end{algorithmic}  
\end{algorithm}  

\vspace{-2ex}
\subsection{\algoGreedy Algorithm}\label{subsec:gr}
The drawback of the {\algoDistancefirst} algorithm is that the rider
request with the earliest submission time is selected in each
iteration, which neglects the served rate.
Therefore, we propose a {\algoGreedy} algorithm in
Algo.~\ref{gr:main} to deal with our dual-objective
problem: maximize the served rate and minimize the additional
distance. 

The {\filter} phase in {\algoGreedy} is similar to that in
{\algoDistancefirst}.
However, in the {\refine} phase, we adopt a dispatch strategy by
allocating the rider to the driver that results in the best utility
gain (i.e., minimum utility value) locally each time, until there is
no remaining driver-rider pair to refine.
Since we want to make more riders happy, where the served
riders can have less additional distance, we use a utility function
to represent the average additional distance of the served riders. 
Then if $r_j$ is served by $d_i$, the utility function {\utilitygr}
is defined as:
\vspace{-0.5ex}
\begin{equation}\label{gr:utility}
\text{\utilitygr}=\frac{AD}{\text{\rider}.rn}
\end{equation}

A rider request attached with more riders ($\rider.rn \textgreater
1$) is more favoured according to our utility gain
(Eq.~\ref{gr:utility}).
For example, if a rider makes a ridesharing request with a friend,
they tend to have the same source and destination point, which
implies that the detour distance and waiting time can be reduced or
even avoided (if not shared with other riders).
On the contrary, if two separate requests are served by a driver, it
is imperative for us to coordinate the dispatch permutation sequence
and guarantee each of them is satisfied.
We expect that $U$ is minimized as much as possible, i.e., more
riders are served with less detour distance.
\setlength{\algomargin}{0.5em} 
\setlength{\textfloatsep}{0.1cm}
\begin{algorithm}[t]
	\footnotesize
	\caption{\small The Greedy Algorithm (GR) }   
	\label{gr:main}   
	\begin{algorithmic}[1] 
		\REQUIRE  
		a driver set $D$, a rider request set $R$
		\ENSURE  
		assigned driver-rider pair list $L$;  
		\STATE $P \leftarrow$ a min-priority queue, $L \leftarrow \emptyset$   						\label{gr:initialize}
		\FOR {$G_u \in G$ }																	\label{gr:for:start}
		\FOR {\rider $ \in R$ }												\label{gr:for1}
		\IF[Lemma \ref{lemma-one}] {$dis^{\downarrow}($\rider.\source$,G_u)\leq$\rider.$w$}	\label{gr:for:if}
		\FOR {\driver $ \in G_u.D_u$ }														\label{gr:for:if:for:start}
		\IF[Lemma \ref{lemma-two}] {$dis^{\downarrow}$(\rider.\source,\driver.$l$)$\leq$\rider.$w$}	\label{gr:for:if:for:if}
		\STATE $\mathscr{U}_{i,j} \leftarrow$ RiderInsertion(\driver, \rider, \utilitygr)	\label{gr:for:if:for:if:compute}
		\IF { $\mathscr{U}_{i,j}\neq \infty$}									\label{gr:for:if:for:if:if}		    				
		\STATE $P.push($\driver,\rider,$\mathscr{U}_{i,j}$)												        \label{gr:for:push}	
		\ENDIF
		\ENDIF																				\label{gr:for:end}
		\ENDFOR
		\ENDIF
		\ENDFOR  
		\ENDFOR
		\label{gr:sorting}
		\WHILE {$P \neq \emptyset$} 											\label{gr:while:start}
		\STATE choose a pair $\langle d, r \rangle$ with minimal utility value from $P$	     	\label{gr:while:choose}
		\STATE insert rider $r$ to the current trip schedule of $d$  				\label{gr:while:if:insert}	
		\STATE $L.push(d, r)$																\label{gr:while:if:push}
		\STATE remove pairs $\langle *, r \rangle$ from $P$												\label{gr:while:if:remove}
		\FOR {($d$, \rider) $ \in$ ($d$, *) }												\label{gr:while:if:for:start}
		\IF {RiderInsertion($d$, \rider, \utilitygr) $\neq \infty$}												\label{gr:while:if:for:if:feasible}
		\STATE update the utility value of $\langle d, \rider \rangle$ in $P$ 		    	\label{gr:while:if:for:if:update}	
		\ELSE
		\STATE remove the pair $\langle d, \rider \rangle$ from $P$										\label{gr:while:if:for:else:remove}
		\ENDIF
		\ENDFOR 																			\label{gr:while:if:for:end}
		\ENDWHILE																			\label{gr:while:end}
		\RETURN $L$ 
	\end{algorithmic} 
\end{algorithm}

The {\algoGreedy} method is presented in Algo.~\ref{gr:main}.
We initialize a min-priority queue $P$ to save the valid driver-rider
pair candidates with their utility values, and an empty set $L$ to store
our final result (line~\ref{gr:initialize}).
In the {\filter} phase, we traverse each driver-rider pair to check
whether they can be matched
(lines~\ref{gr:for:start}-\ref{gr:for:end}).
If the driver can carry the rider, then we will calculate its utility
value (line~\ref{gr:for:if:for:if:compute}) and push it into a pool
$P$ (line~ \ref{gr:for:push}).
In the {\refine} phase, in each iteration we select the pair $\langle
d,r \rangle$ with the smallest utility value greedily
(line~\ref{gr:while:choose}), we perform an insertion
(line~\ref{gr:while:if:insert}), and append the pair into our result
list $L$ (line \ref{gr:while:if:push}).
Meanwhile, we remove all pairs related to $r$ from $P$
(line~\ref{gr:while:if:remove}).
For riders where $d$ was considered as a candidate, we check
whether it is still valid to include them as a pair since the
insertion of a new rider may influence the riders previously
considered
(lines~\ref{gr:while:if:for:start}-\ref{gr:while:if:for:end}).
If $d$ is still feasible, we update the utility gain value for those
riders (line~\ref{gr:while:if:for:if:update}).
Otherwise, the driver-rider pair is removed from $P$
(line~\ref{gr:while:if:for:else:remove}).

\textbf{Time complexity analysis.}
Firstly, the time complexity to traverse each
driver-rider pair to check whether they can be matched is
$O(\tau|R||D_u|n^2\shortestpathquery)$.
 Each time we select a driver-rider pair with the lowest utility
value greedily to insert ($O(|P|log|P|)$), remove those pairs related
to $r$ from $P$ ($O(|P|)$) and then update the other influenced
riders ($O(|P|n^2\shortestpathquery)$). Hence, the total time complexity is
$O(\tau|R||D_u|n^2\shortestpathquery$+$|P|^2log|P|$+$|P|^2$+$|P|^2n^2\shortestpathquery)$.

\vspace{-10px}
\subsection{Updating Process}\label{subsec:update}
After every {\updatetimeinter} time interval, we carry out the update
process, including updating each non-empty driver's current location,
trip schedule and index information.
If a driver is empty, then we assume that the state is static or
inactive and unnecessary to be updated.
\vspace{-2px}
\subsubsection{Update of the driver's current location} 
We first obtain the trip schedule segment where the driver is located
in a coarse-grained way, such as $(o_k,o_{k+1})$ according to the
actual completion distance $\dist(o_0,o_k)$ of $o_k$ and the vehicle
travel distance ($\speed \times \updatetimeinter$) within one
timeslot, which has an $O(n)$ time complexity.
Then we get the traversing edges set $E_k$ following the shortest
path between $o_{k}$ and $o_{k+1}$, such as
$e_{k1},e_{k2},\ldots,e_{kk}$.
According to the remaining travel distance
($\speed \times \updatetimeinter - \dist(o_0,o_k)$), we go through
each edge in $E_k$ and pinpoint the edge $e_d$ where the driver is
exactly, which has $O(|E_k|)$ time complexity.
Finally, we choose an endpoint of $e_d$ as the current location of
the driver approximately.
In the worst case, the driver location update takes $O(n+|E_k|)$.
\vspace{-1.5ex}
\subsubsection{Trip schedule updates with slack distance}
If the vehicle picks up new riders or drops off passengers on board,
then the corresponding source or destination point should be removed
from the original trip schedule.
Meanwhile, we update the slack distance of each point $o_k$, which
takes $O(n)$ time.
\vspace{-1.5ex}
\subsubsection{Update of the road network index information}
Each subgraph $G_i$ maintains a list of vehicles which belong to
$G_i$.
If the driver moves from $G_i$ to another subgraph $G_j$, then it
will be removed from $G_i$ and inserted into $G_j$.
Given a vertex, it only takes $O(1)$ to obtain its situated subgraph.
Thus updating the road network index takes $O(|D|)$ time.
\vspace{-10px}

	\section{Experiments}\label{sec:exp}
In this section, we perform an experimental evaluation on our
proposed approaches.
We first present the experimental settings in
Section~\ref{subsec:exp:setting}, and then validate the efficiency
and effectiveness in Section~\ref{subsec:exp:result:real} and
Section~\ref{subsec:exp:result:case}.
\vspace{-5px}
\subsection{Experimental Setting}\label{subsec:exp:setting}

\textbf{Datasets.}
We conduct all experimental evaluations on a real \textit{Shanghai}
dataset~\cite{huang2014large}, which includes road network data and
taxi trips.
The detailed statistical information is as follows: 

\begin{itemize}[leftmargin=*]
\item \textit{Road Network.}
There are $122{,}319$ vertices and $188{,}426$ edges in the
\textit{Shanghai} road network dataset~\cite{huang2014large}.
Each vertex contains the latitude and longitude information.
For each directly connected edge, the travel distance is given.
The shortest path distance between any two vertices can be obtained
by searching an undirected graph of the road network.

\item \textit{Taxi Trip.}
In the \textit{Shanghai} taxi trip dataset~\cite{huang2014large},
there are $432{,}327$ taxi requests on May 29, 2009.
Each taxi trip records the departure timestamp, pickup point, and
dropoff point.
Although we cannot estimate the exact request timestamp, we use the
departure timestamp to simulate the request submission timestamp.
The pickup and dropoff points are pre-mapped to the closest vertex on
the road network.
The initial location of a driver is randomly initialized as a vertex
on the road network.
For all initialized vehicles, we assume that there is no rider
carried at the beginning.
\end{itemize}

\noindent\textbf{Parameter settings.}
The key parameters of this work are listed in Table~\ref{table:parameter}, where the
default values are in bold.
For all experiments, a single parameter is varied while the rest are
fixed to the default value. 
Note that speed {\speed} is a constant value set to
48 km/hr by default for convenience, but can be any value, as
``speed'' in this context is an average over the entire trip.
Even though travel speed varies in reality, the travel time can be
easily obtained if the speed and travel distance are known.
Therefore, the scenario of varying speed is orthogonal to our
problem.

\noindent\textbf{Implementation.}
Experiments were conducted on an Intel(R) Xeon(R) E5-2690 CPU@2.60GHz
processor with 256GB RAM.
All algorithms were implemented in C++.
The experiments are repeated $10$ times under each experimental
setting and the average results are reported.
For the ``matching time'' and ``update time'' metrics, we also report
the results distribution.
The size of our road network index is 1.85GB, and the index used in \dsa is 2GB.

\noindent\textbf{Compared algorithms.}
The following methods were tested in our experiments:
\begin{itemize}[leftmargin=*]
\item {\dsa}~\cite{chen2018priceAndtime}.
The primary baseline in our experiments, which has been shown to
outperform another state-of-the-art approach
recently~\cite{huang2014large}.
\item {\df}.
The distance-first algorithm described in Section~\ref{subsec:df}. 
\item {\gr}.
The greedy algorithm described in Section~\ref{subsec:gr}.
\item {\dfp}.
The distance-first algorithm described in Section~\ref{subsec:df}
with our proposed pruning rules in Section~\ref{sec:riderinsertion}.
\item {\grp}.
The greedy algorithm described in Section~\ref{subsec:gr} with our
proposed pruning rules in Section~\ref{sec:riderinsertion}.
\end{itemize}

\noindent\textbf{Evaluation Metrics.}
We evaluate both \textit{efficiency} and \textit{effectiveness} under
different parameter configurations.
\begin{figure*}[t]
	\centering
	\subfloat[Served rate\label{wait_servedRate}]{
		\includegraphics[width=5.34cm]{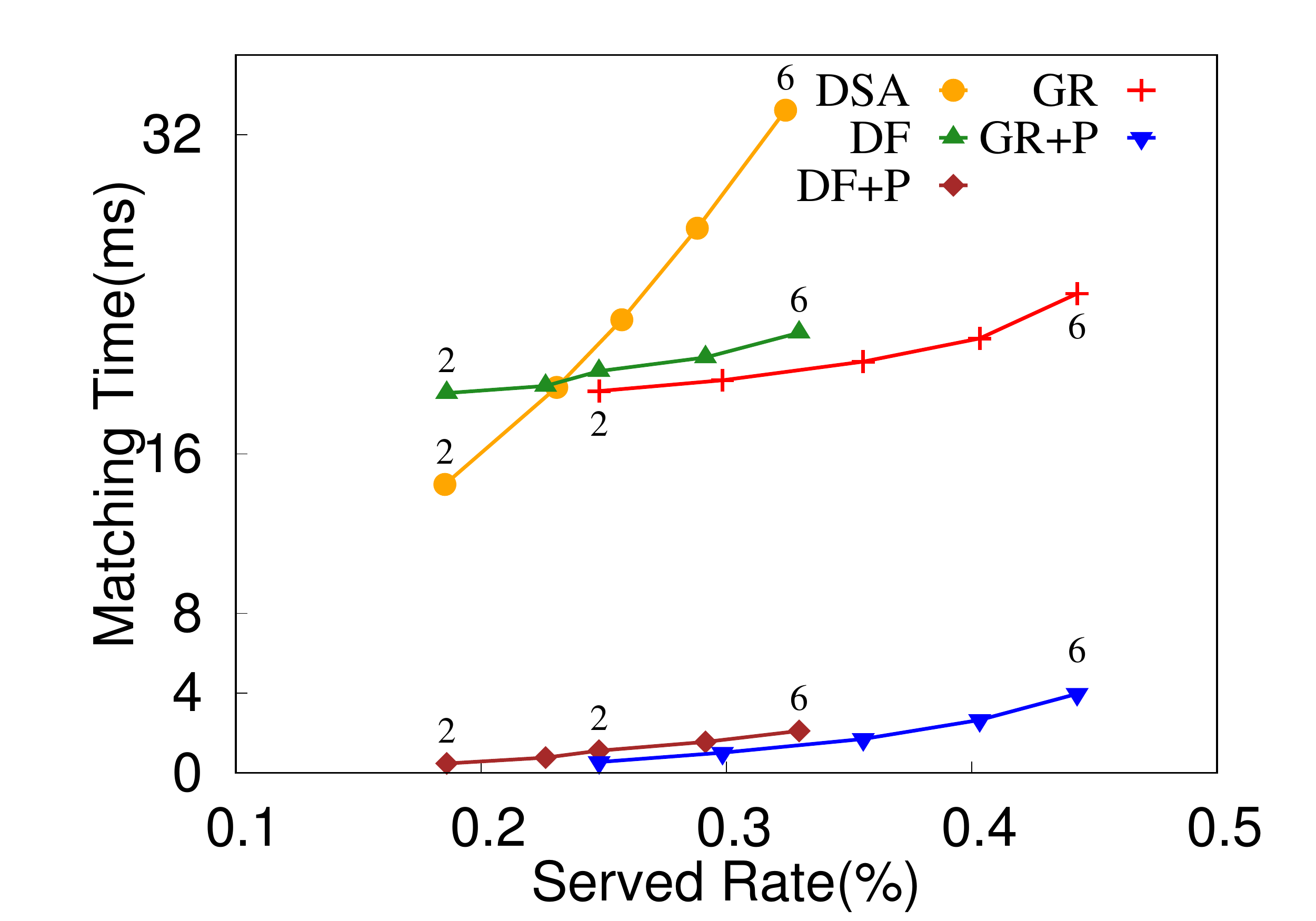}}
	\hspace{1px}
	\subfloat[Additional distance\label{wait_dis}]{
		\includegraphics[width=5.34cm]{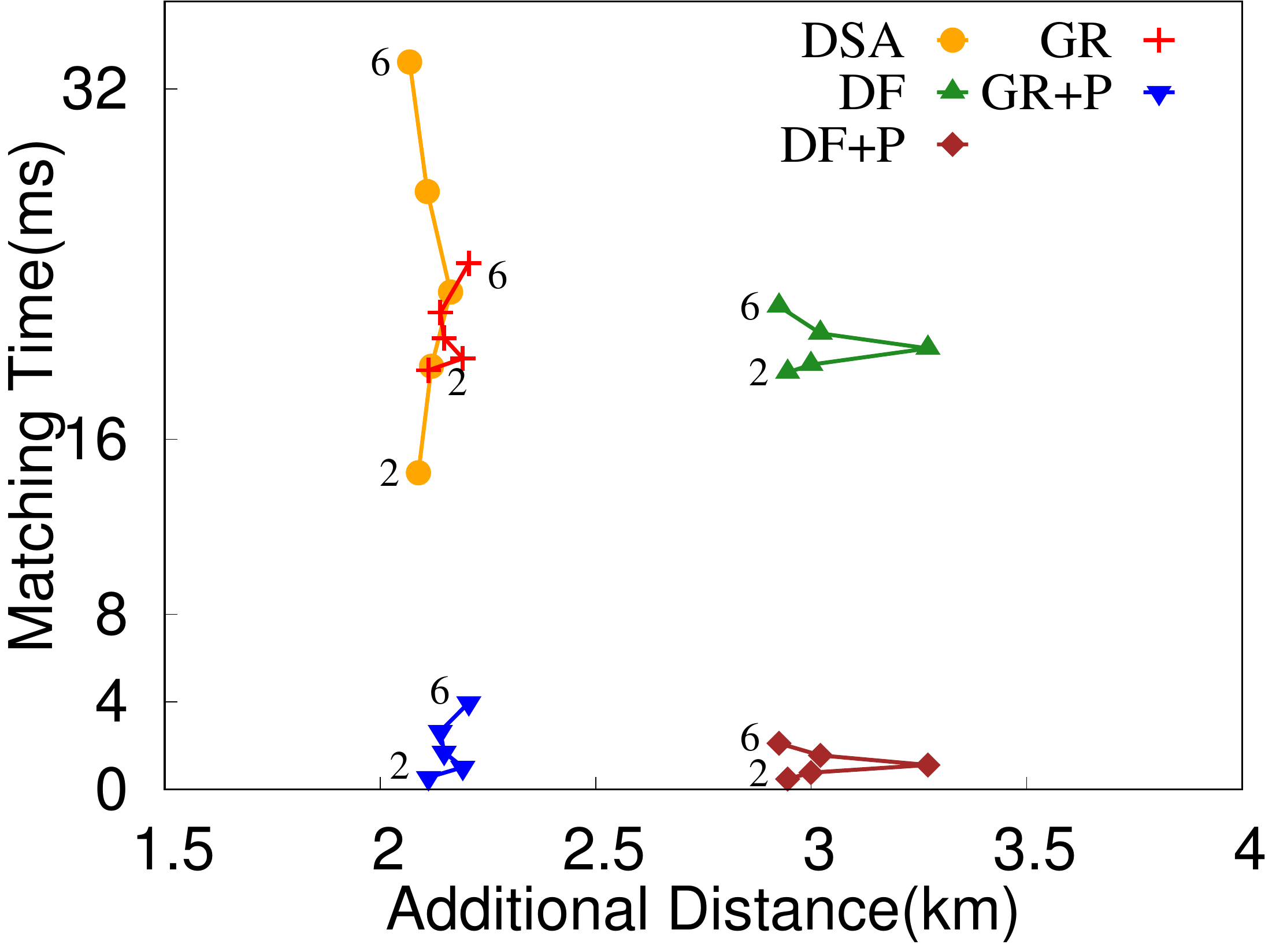}}
	\hspace{1px}
	\subfloat[Update time (The update time of \dsa is over $10$s, we omit it in the plot for better visualization)\label{wait_update}]{
		\includegraphics[width=5.34cm]{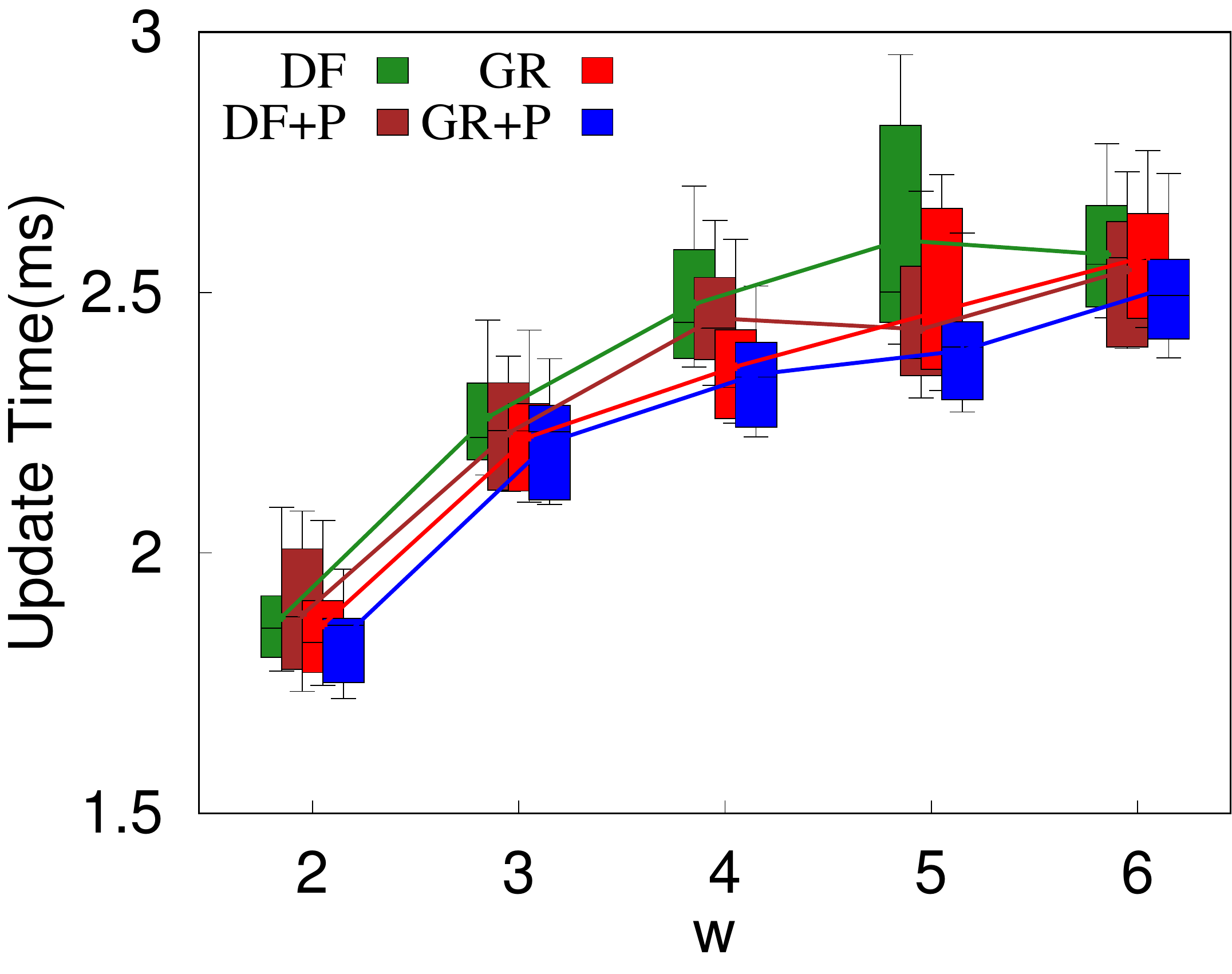}}
	\vspace{-1ex}
	\caption{Performance when varying the waiting time constraint {\wait} from $2$ to $6$ minutes}
	\label{fig-wait} 
	\vspace{-10pt} 
\end{figure*}
\begin{figure*}[t]
	\centering
	\subfloat[Served rate\label{seat_servedRate}]{
		\includegraphics[width=5.34cm]{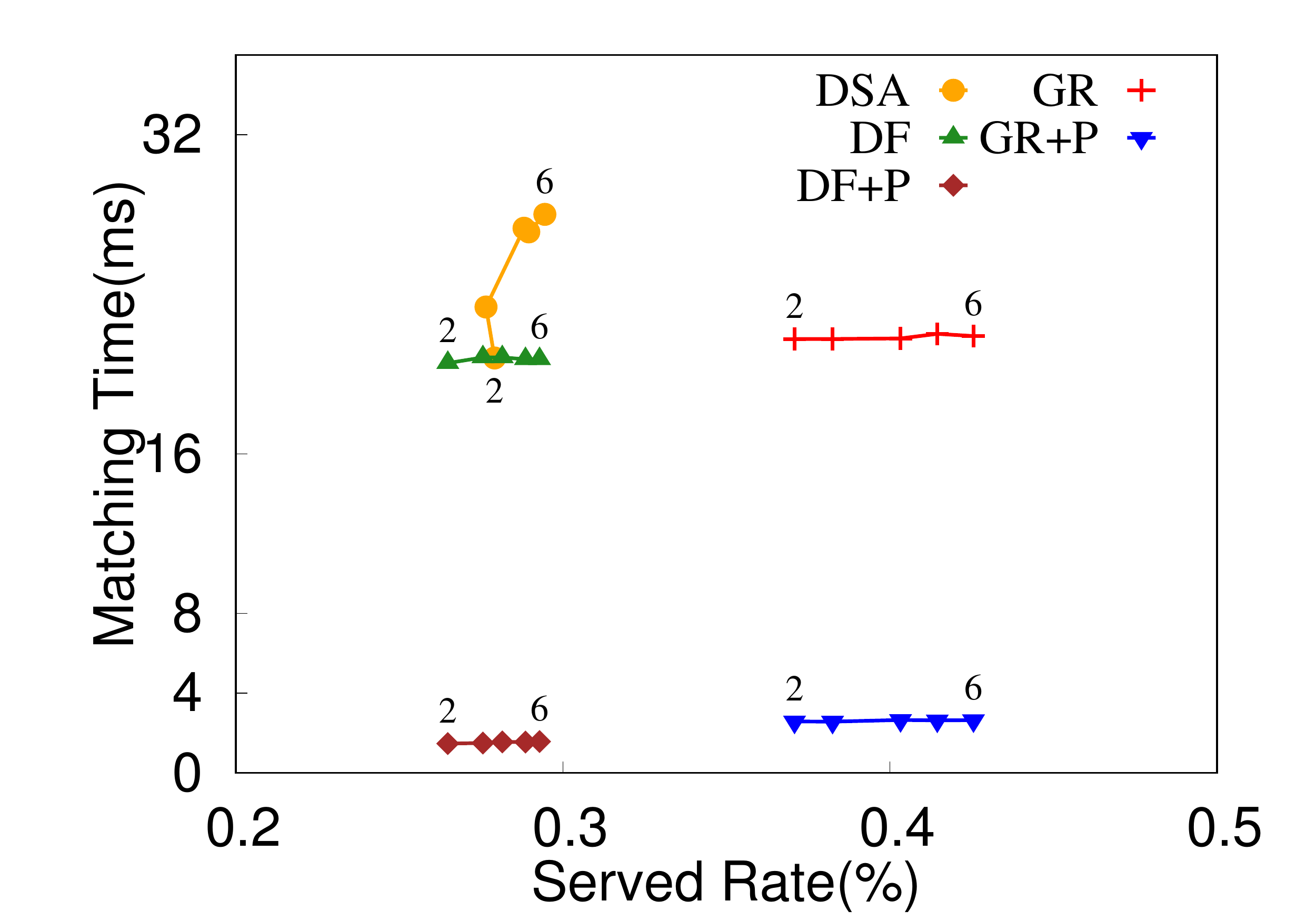}}
	\hspace{1px}
	\subfloat[Additional distance\label{seat_dis}]{
		\includegraphics[width=5.34cm]{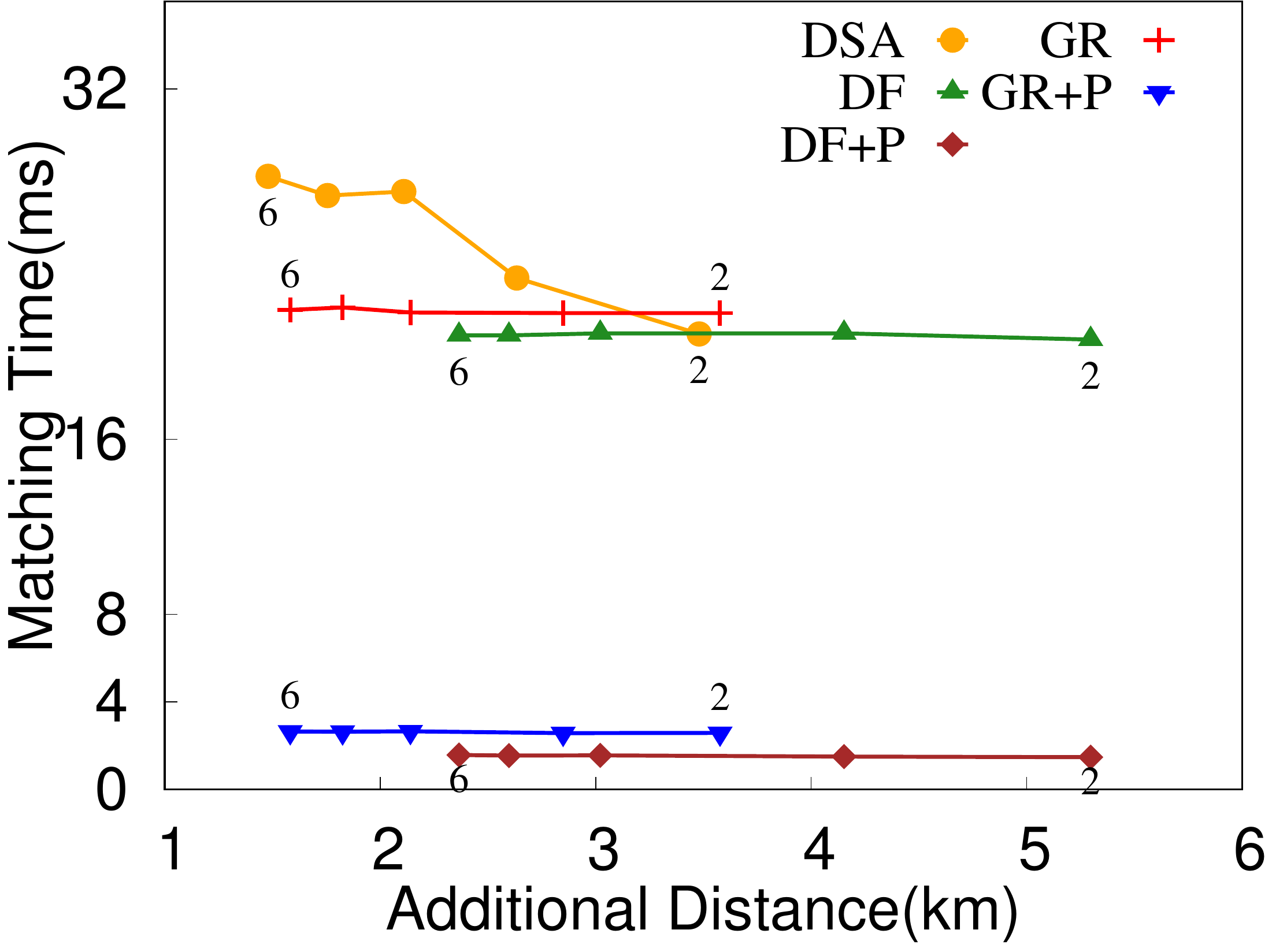}}
	\hspace{1px}
	\subfloat[Update time (The update time of \dsa is over $14$s, we omit it in the plot for better visualization)\label{seat_update}]{
		\includegraphics[width=5.34cm]{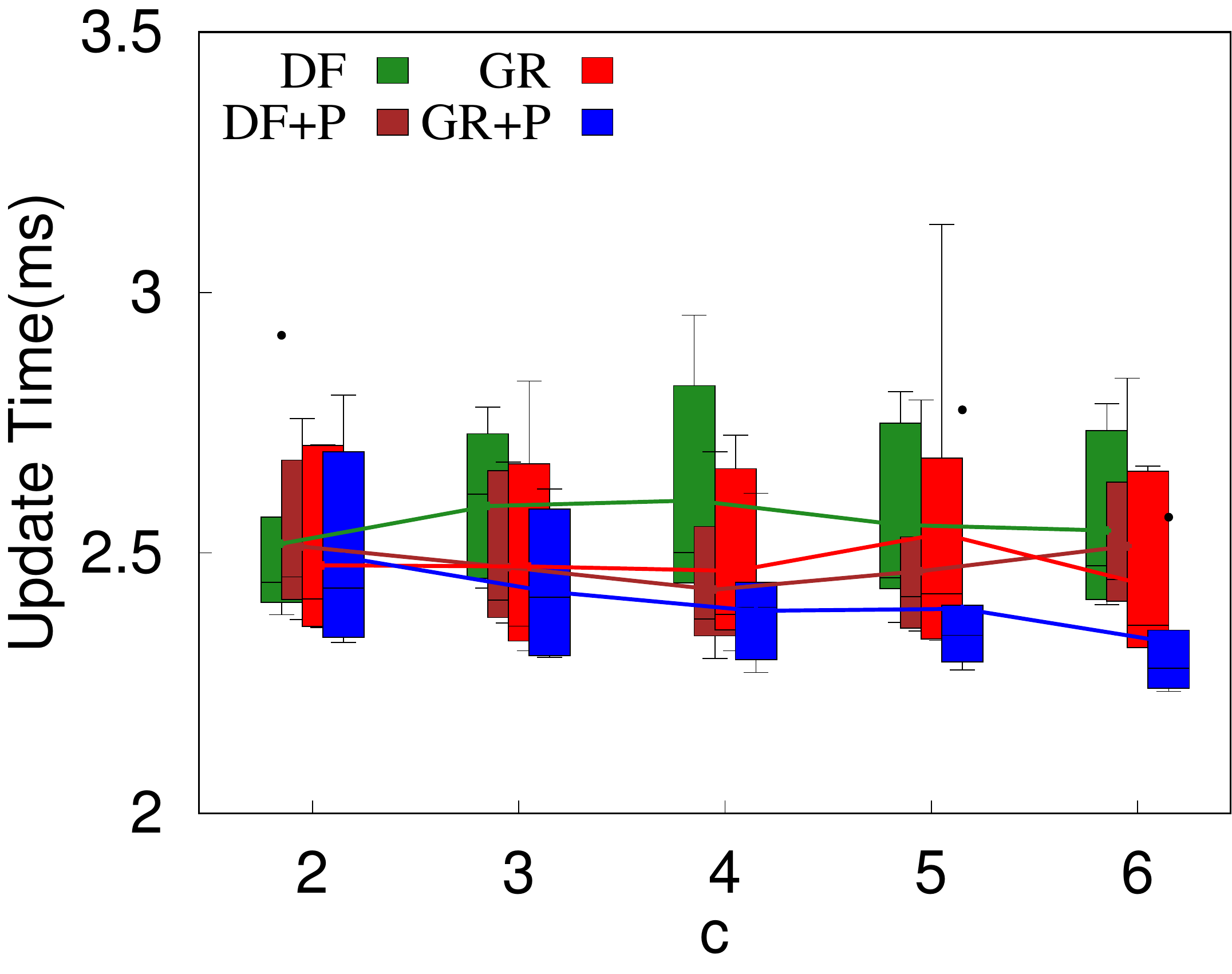}}
	\vspace{-1ex}
	\caption{Performance when varying the capacity {\ms} from $2$ to $6$}
	\label{fig-seat}  
	\vspace{-12px}
\end{figure*}
\begin{figure*}[t]
	\centering
	\subfloat[Served rate\label{driver_servedRate}]{
		\includegraphics[width=5.34cm]{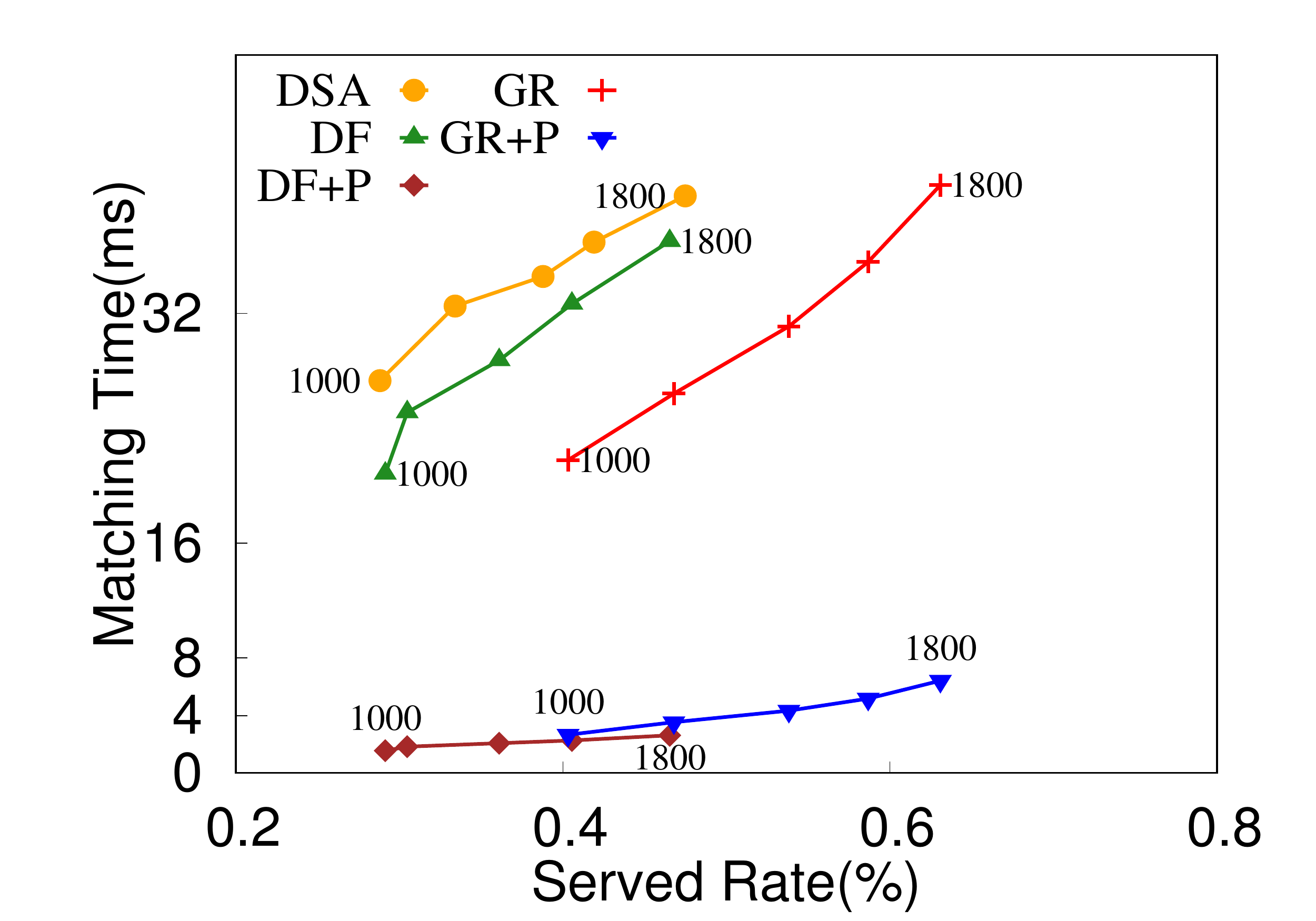}}
	\hspace{1px}
	\subfloat[Additional distance\label{driver_dis}]{
		\includegraphics[width=5.34cm]{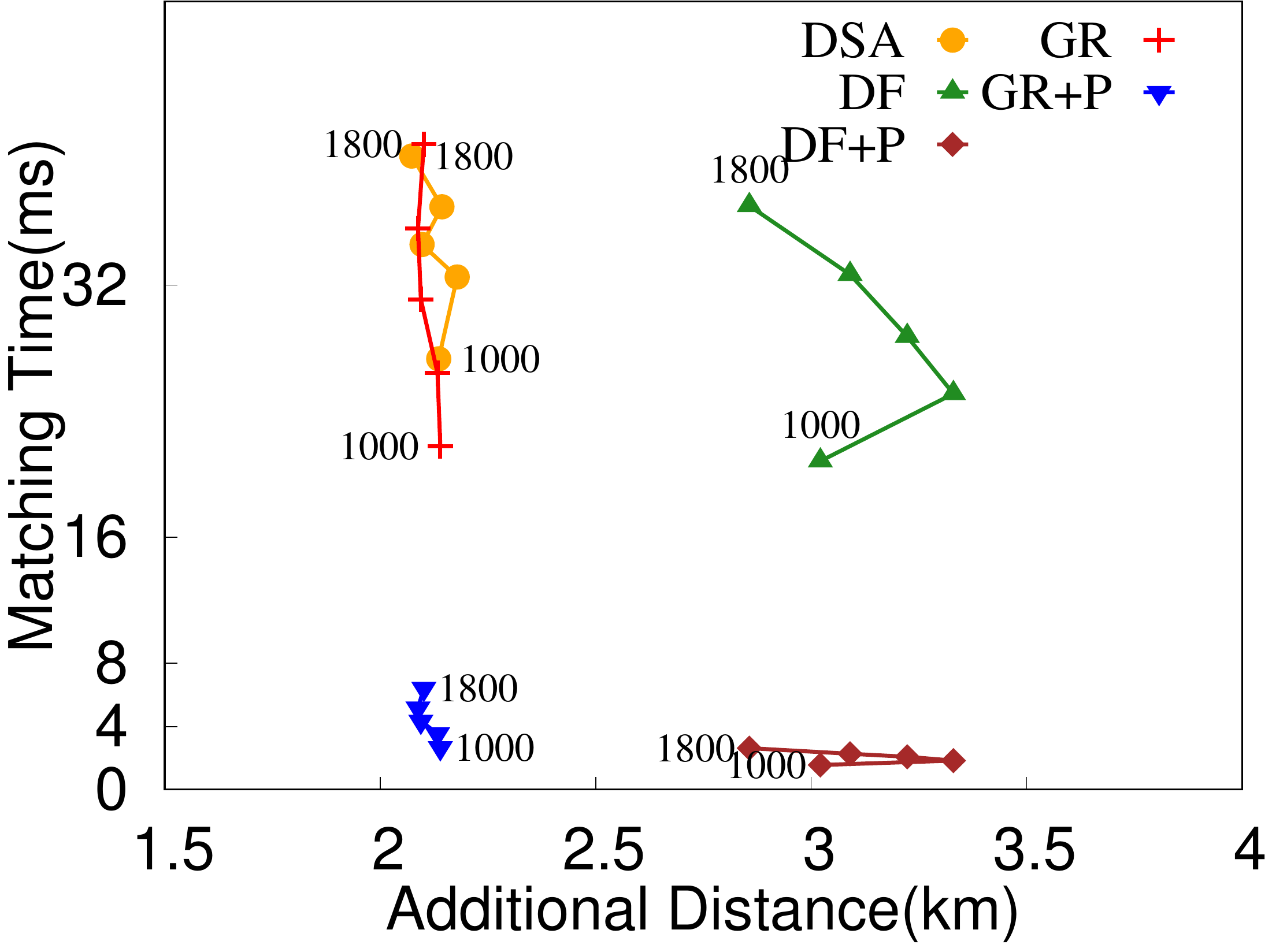}}
	\hspace{1px}
	\subfloat[Update time (The update time of \dsa is over $16$s, we omit it in the plot for better visualization)\label{driver_update}]{
		\includegraphics[width=5.34cm]{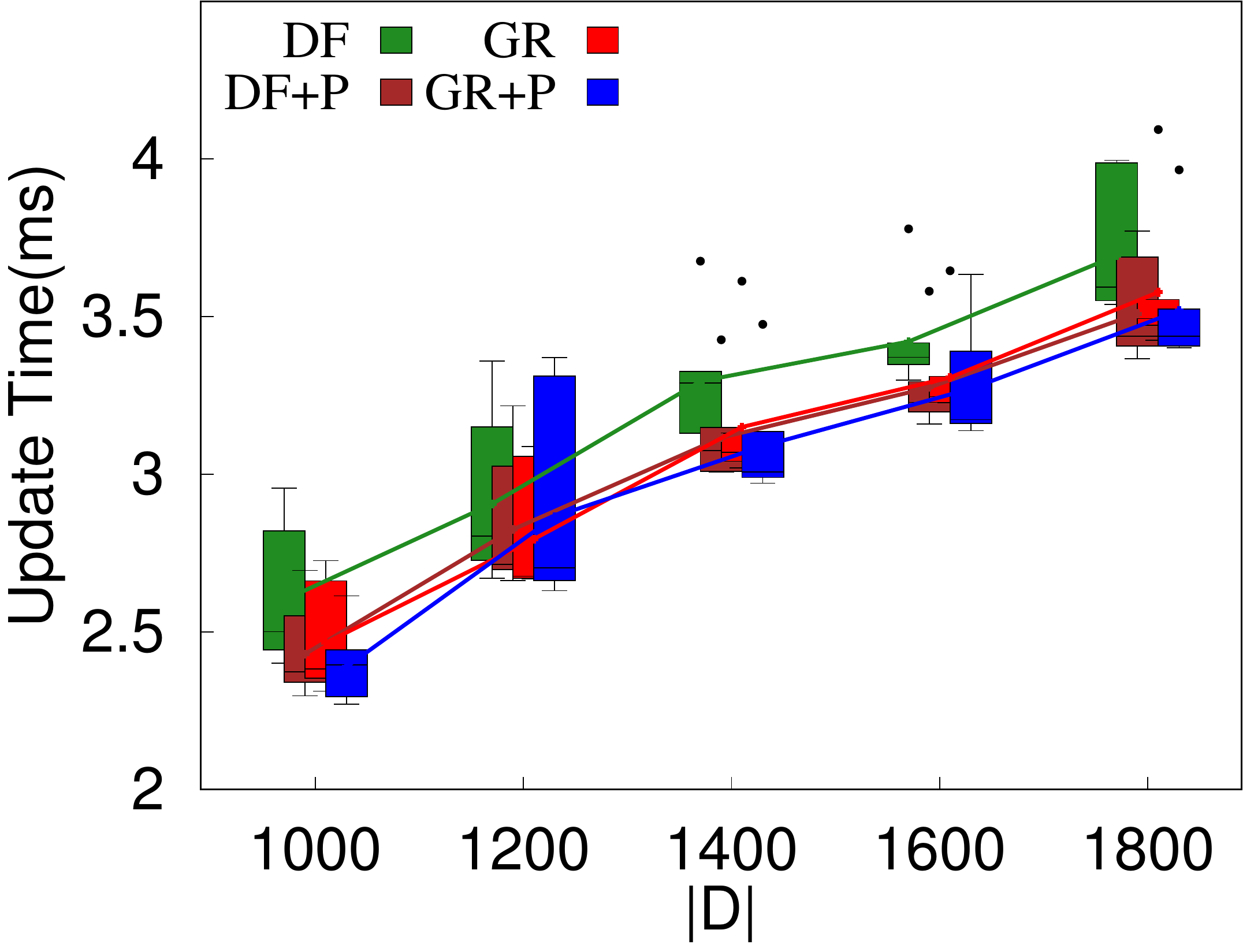}}
	\vspace{-1ex}
	\caption{Performance when varying the number of drivers $|D|$ from $1,000$ to $1,800$}
	\label{fig-driver}
	\vspace{-14px}
\end{figure*}
\begin{figure*}[t]
	\centering
	\subfloat[Served rate\label{speed_servedRate}]{
		\includegraphics[width=5.34cm]{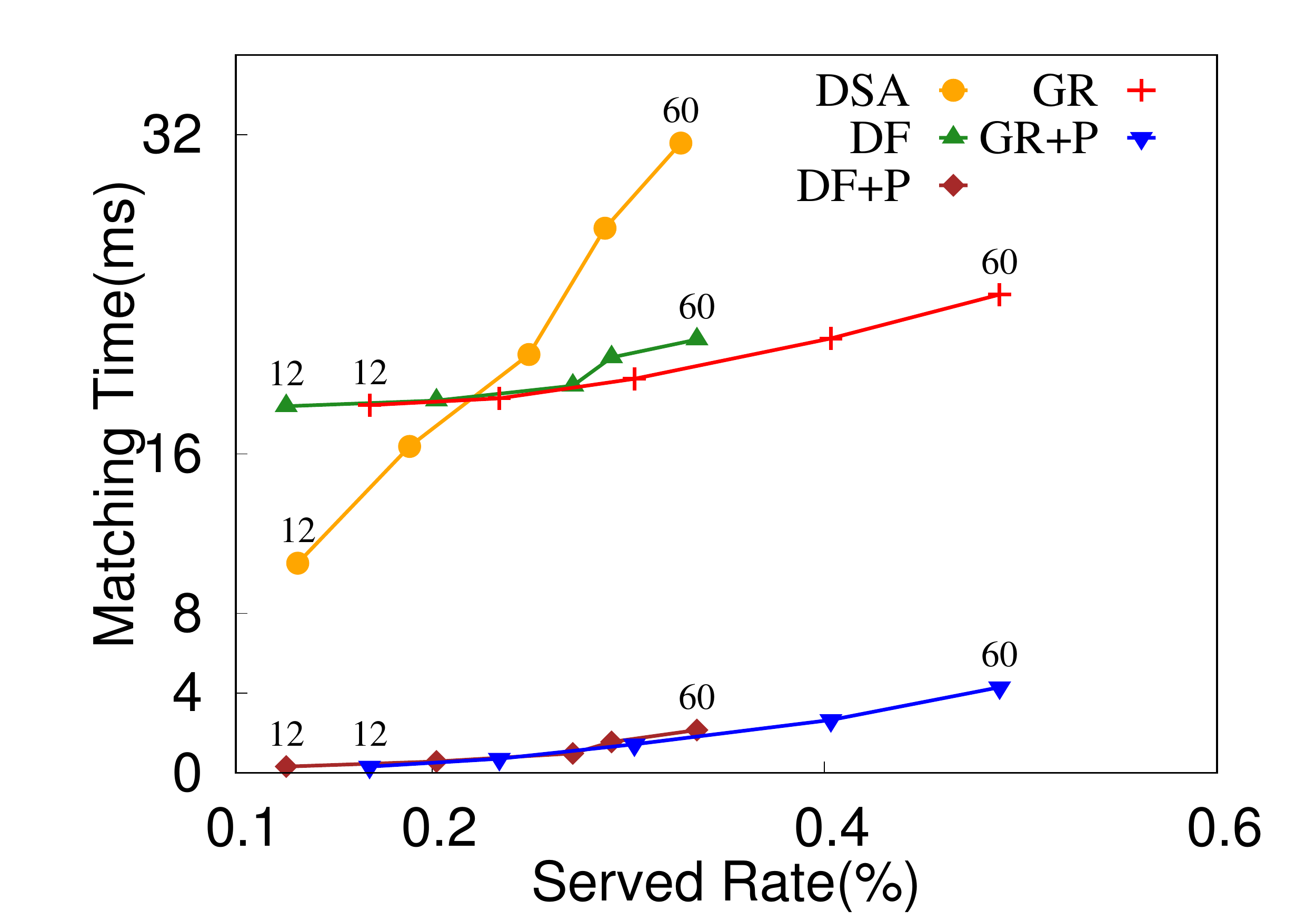}}
	\hspace{1px}
	\subfloat[Additional distance\label{speed_dis}]{
		\includegraphics[width=5.34cm]{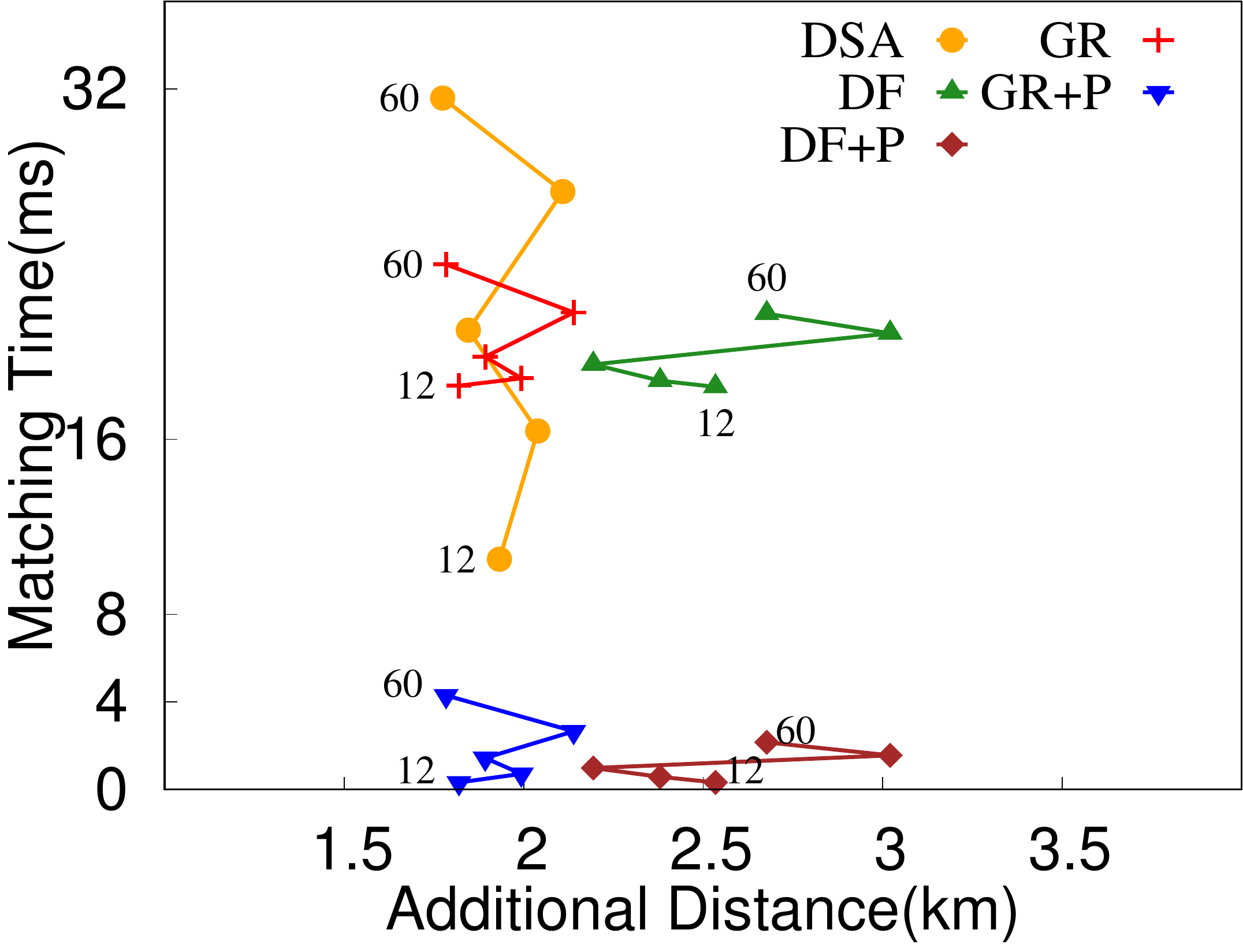}}
	\hspace{1px}
	\subfloat[Update time (The update time of \dsa is over $10$s, we omit it in the plot for better visualization)\label{speed_update}]{
		\includegraphics[width=5.34cm]{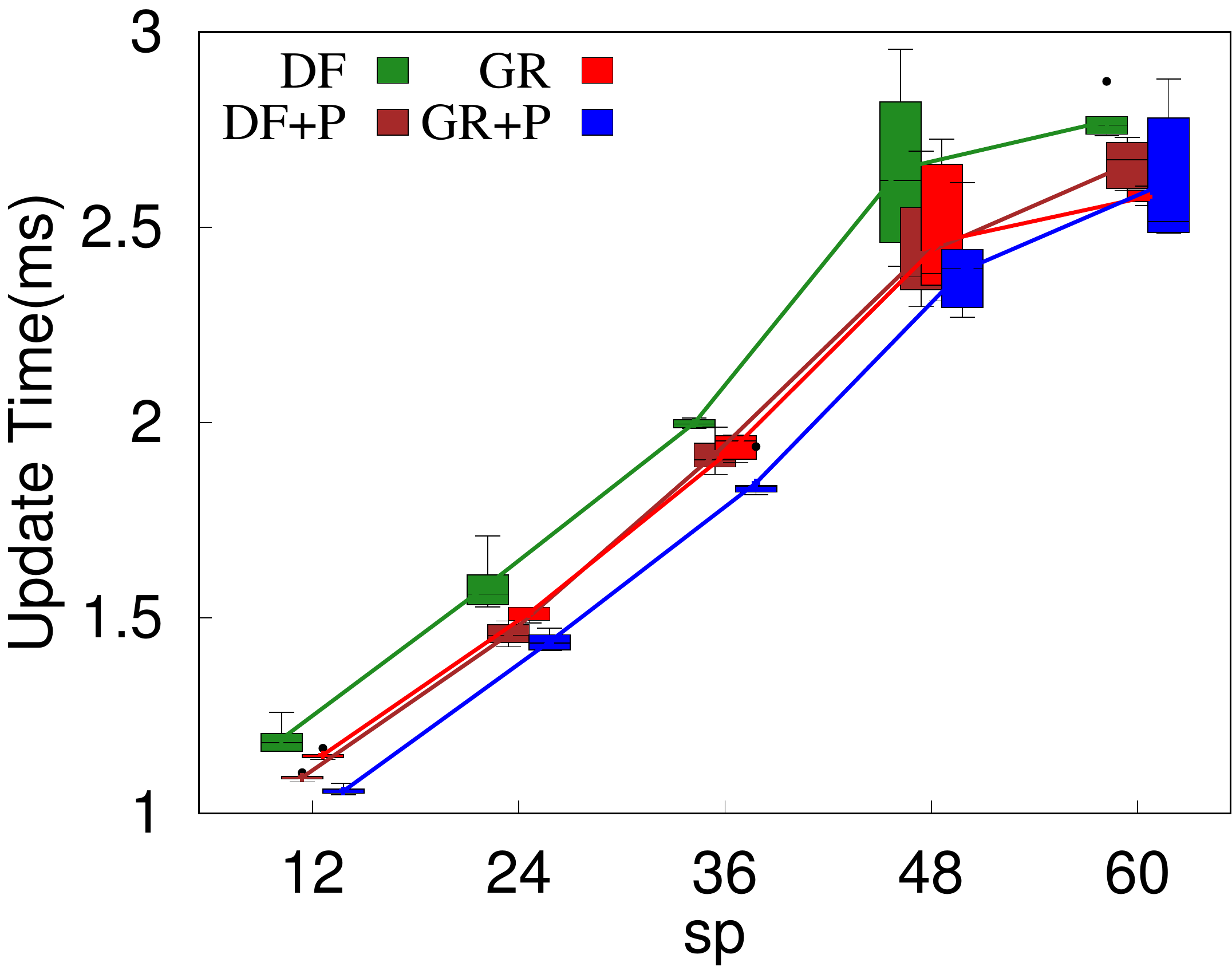}}
	\vspace{-1ex}
	\caption{Performance when varying the travel speed {\speed} from $12$ to $60$ km/hr}
	\label{fig-speed}  
	\vspace{-15px}
\end{figure*}

For the \textit{efficiency}, we compare the \textit{matching time}
and \textit{update time}, where \textit{matching time} represents the
average running time needed to match a driver to a single rider
request, and the \textit{update time} is the average running time
required to amend a driver's current trip schedule and update the
underlying index.
Note that we omit the measurement of the driver's current location
update because the driver's location can be acquired from the GPS
device in near realtime.

For the \textit{effectiveness}, we compare the \textit{served rate} and the
\textit{additional distance}.
The \textit{served rate} denotes the ratio of served rider number
divided by the total rider number, and the \textit{additional
distance} is computed as the total additional distance divided by the
number of served riders.

\vspace{-2.2ex}
\subsection{Experimental Results at Peak Travel Period}\label{subsec:exp:result:real}

In this section, we sample $6{,}000$ orders in peak hours, because
{\dsa} cannot process all the rider requests due to combinatorial
time complexity (explained later).
We show a series of \textit{efficiency} and \textit{effectiveness}
trade-off graphs in order to better compare the performance
differences between all of the algorithms.
Beginning and end sweep values are shown for each line to make it
easier to observe the performance trends for each algorithm.
Note that \df and \dfp (or \gr and \grp) maintain exactly the same
values on ``served rate'' and ``additional distance'' metrics.
As variants of certain algorithms have the similar efficiency
profile, but a different one for effectiveness (thus the use of
trade-off graphs).
For the ``update time'' metric, we have added a very small point skew
to the four lines to improve visualization.
\begin{figure*}[t]
	\centering
	\subfloat[Served rate\label{detour_servedRate}]{
		\includegraphics[width=4.53cm]{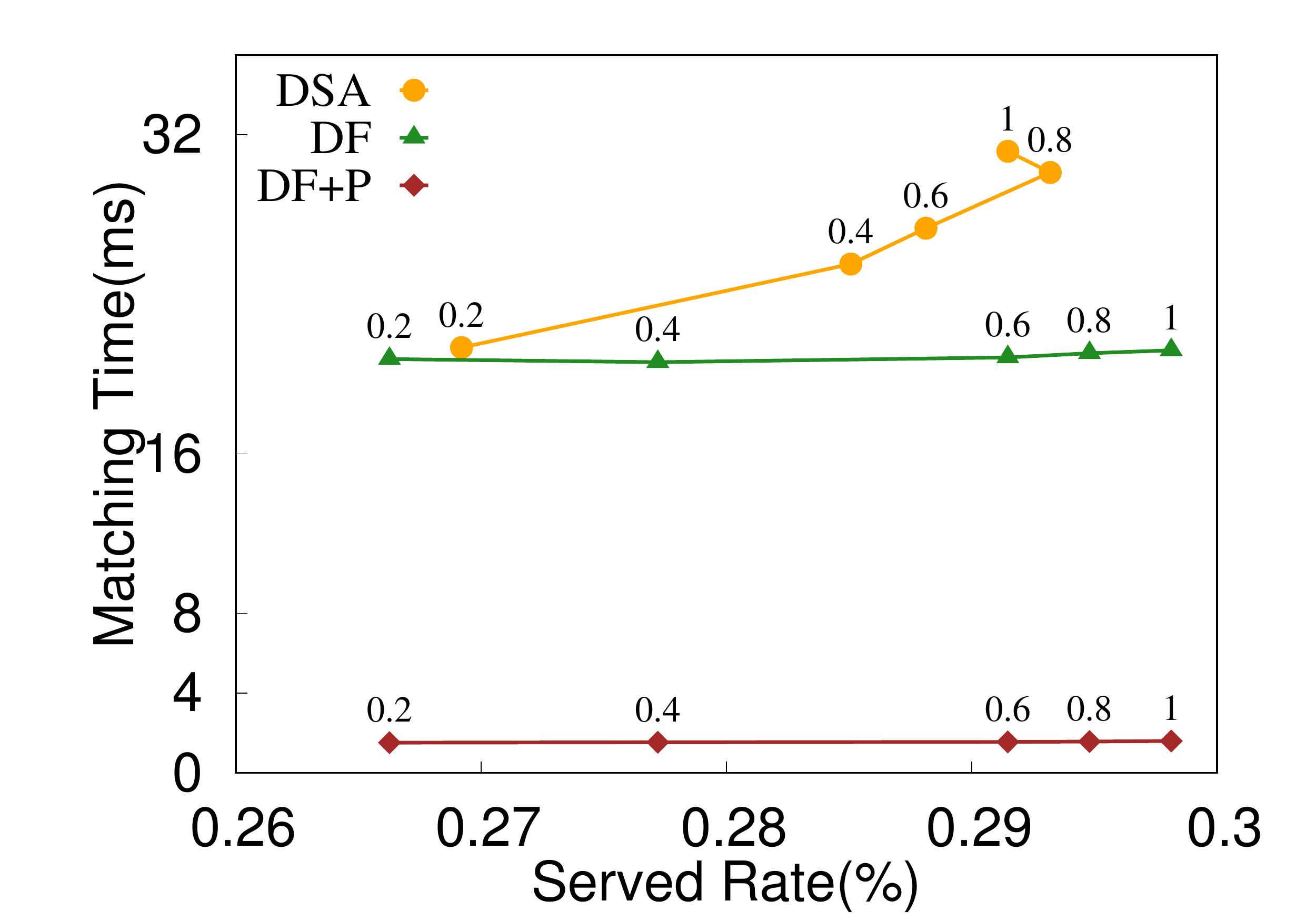}
		\includegraphics[width=4.53cm]{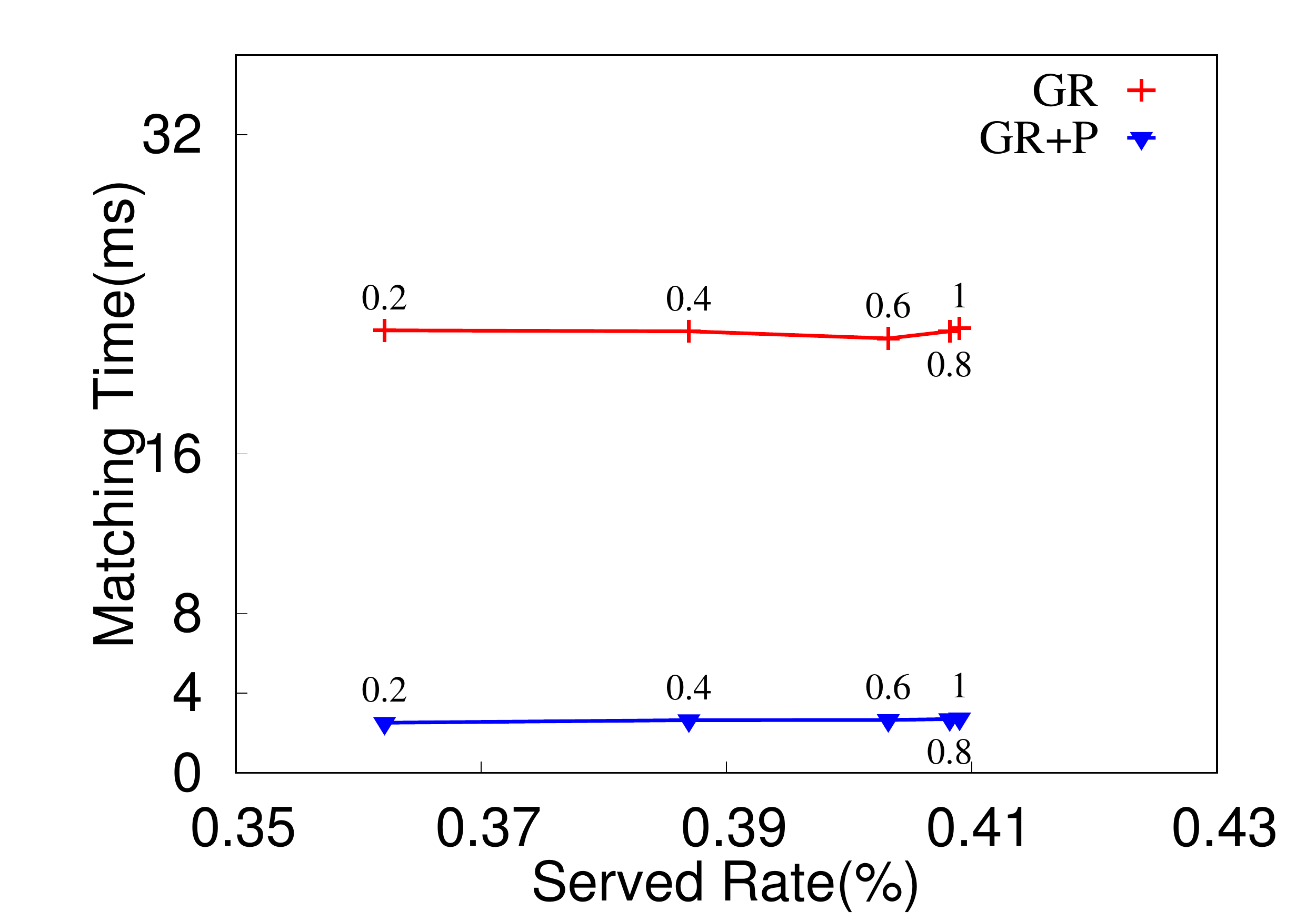}}
	\subfloat[Additional distance\label{detour_dis}]{
		\includegraphics[width=4.53cm]{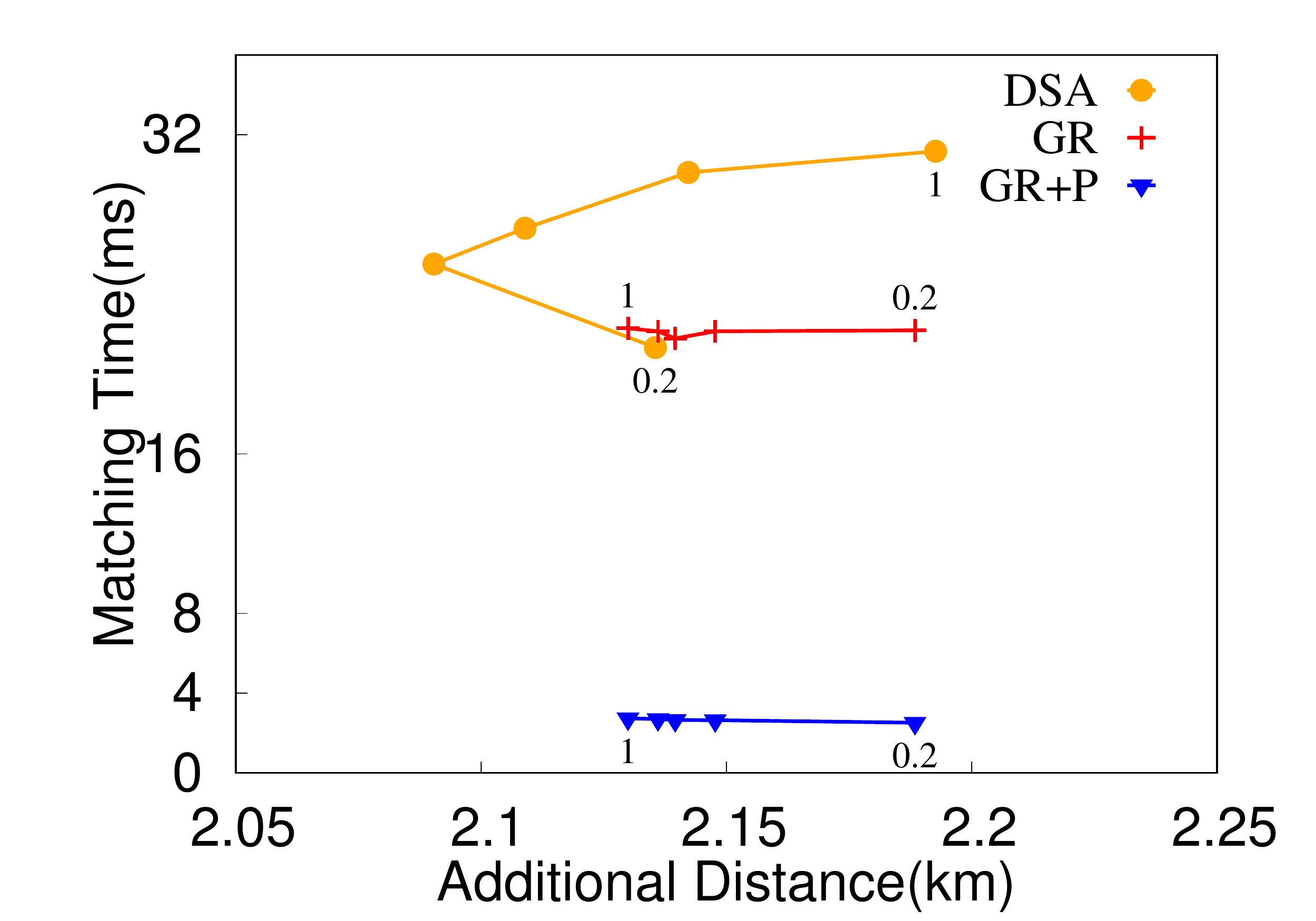}
		\includegraphics[width=4.53cm]{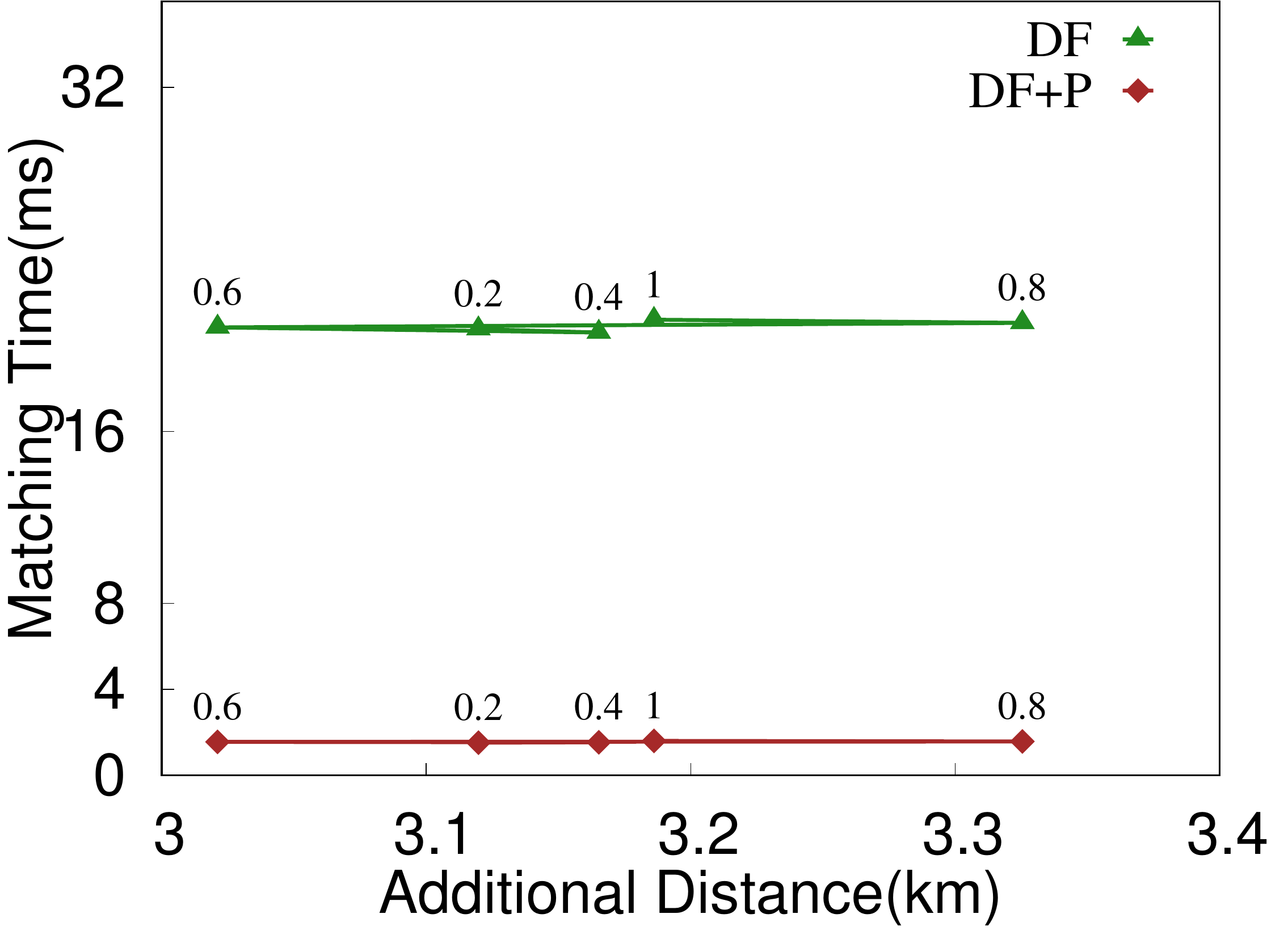}}
	\vspace{-1ex}
	\caption{Performance when varying the detour time constraint {\detour} from $0.2$ to $1$ (The update time is shown in Fig. \ref{detour_update})}
	\label{fig-detour}  
	\vspace{-12px}
\end{figure*}

\begin{figure*}[t]
	\centering
	\subfloat[Served rate\label{window_servedRate_line}]{
		\includegraphics[width=4.53cm]{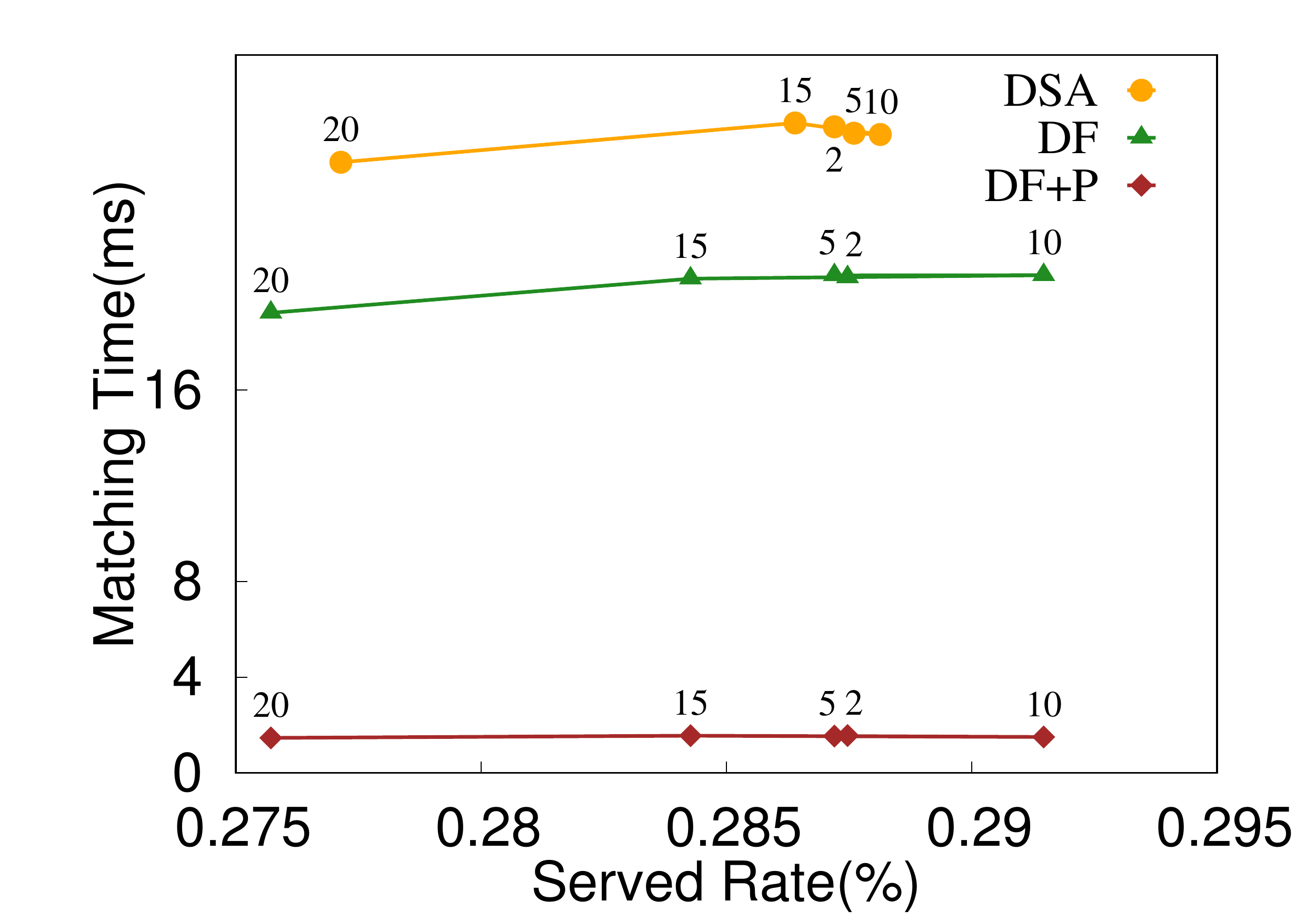}
		\includegraphics[width=4.53cm]{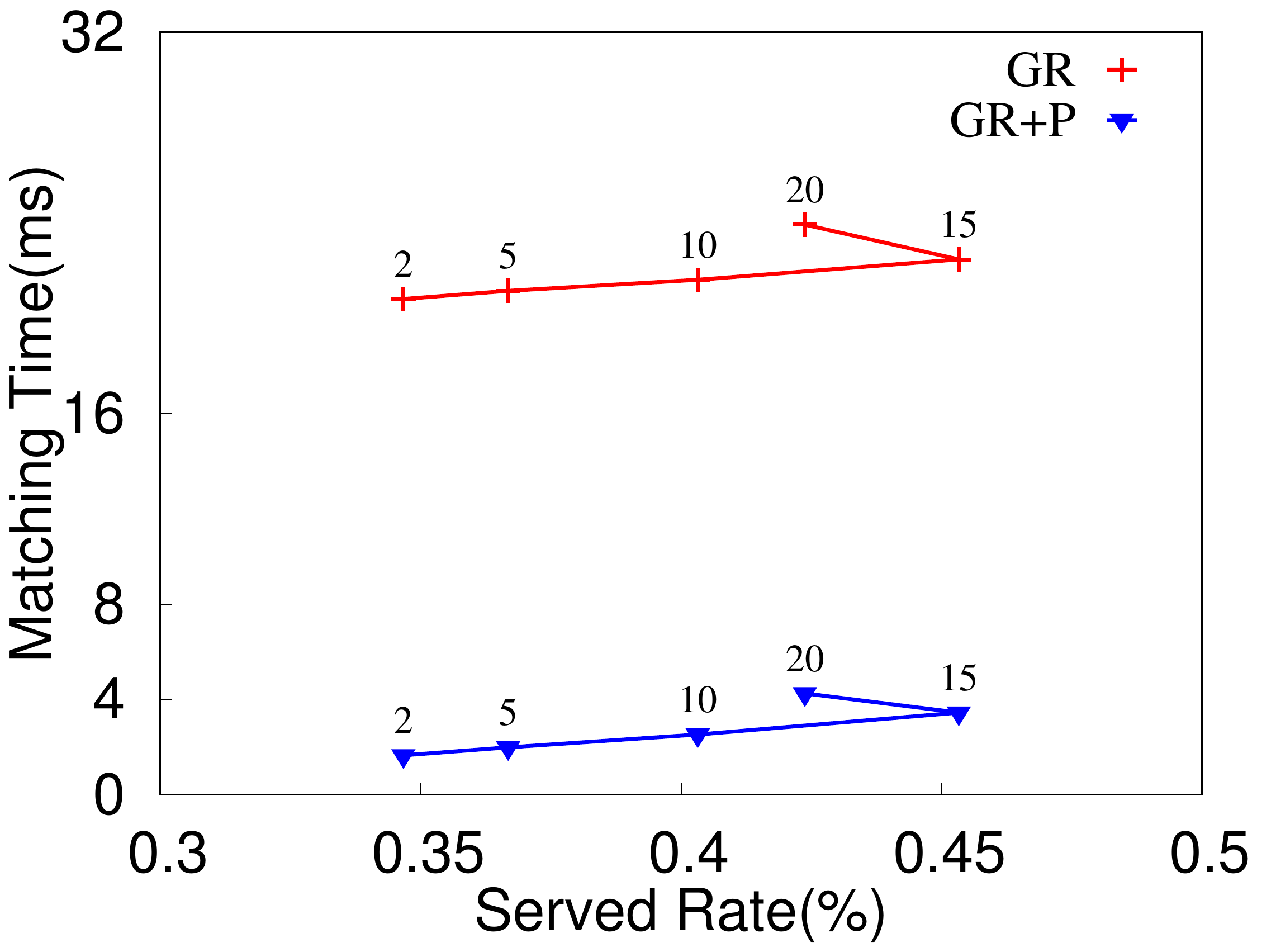}}
	\subfloat[Additional distance\label{window_dis_line}]{
		\includegraphics[width=4.53cm]{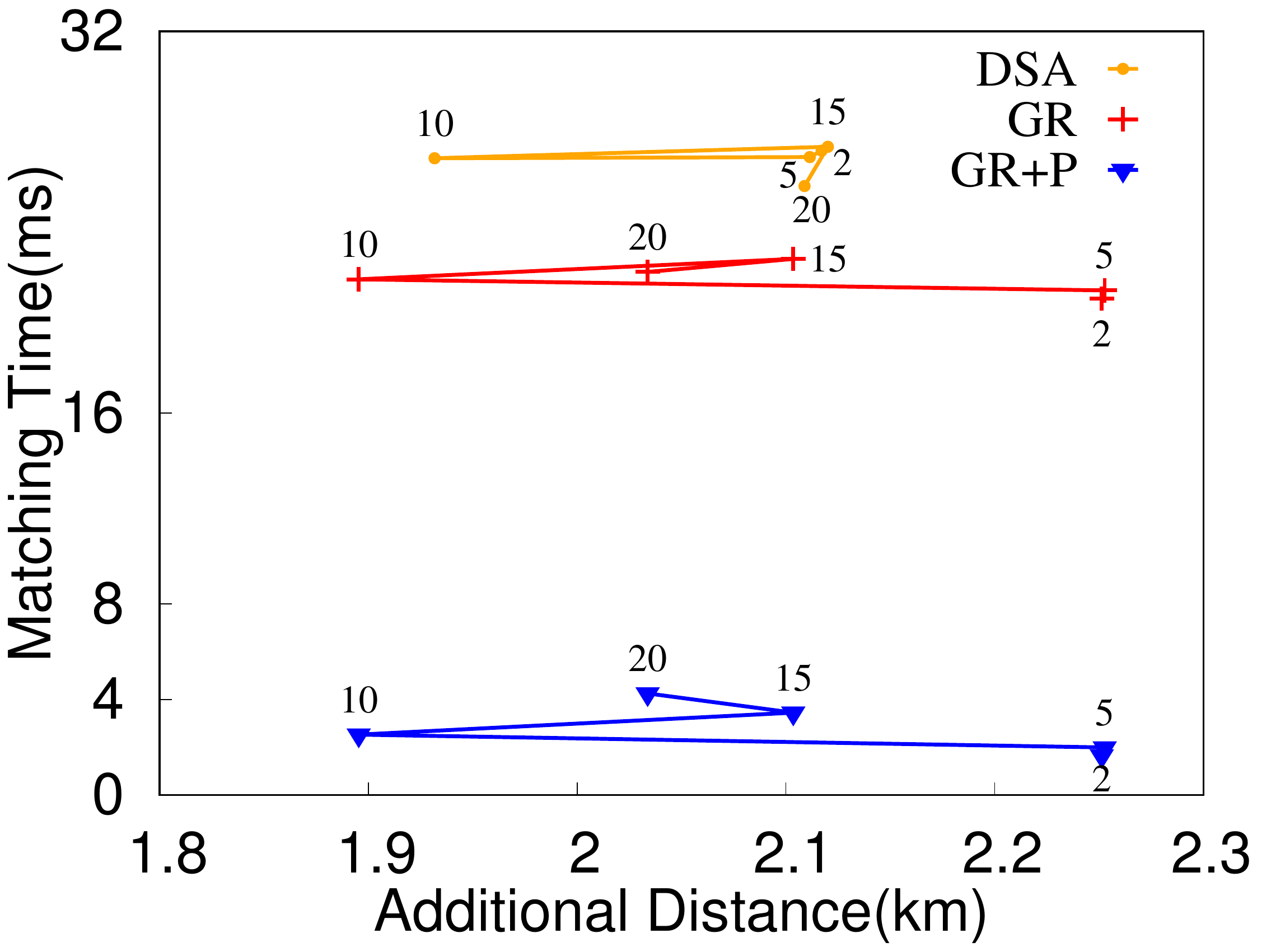}
		\includegraphics[width=4.53cm]{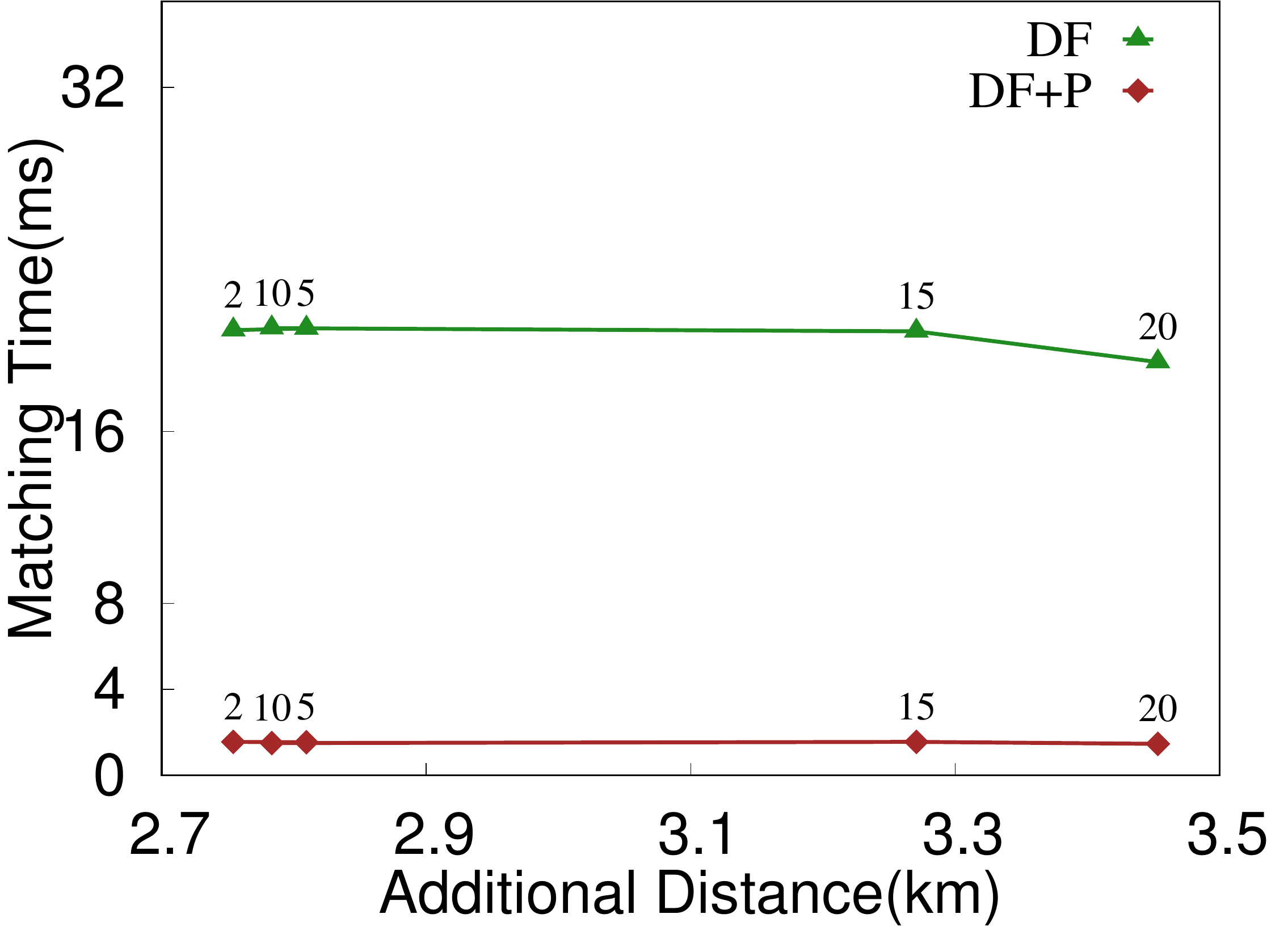}}
	\vspace{-1ex}
	\caption{Performance when varying the update time window {\updatetimeinter} from $2$ to $20$ seconds (The update time is shown in Fig. \ref{window_update})}
	\label{fig-window}  
	\vspace{-10px}
\end{figure*}
\noindent\textbf{Effect of the Waiting Time Constraint $w$. }
Fig.~\ref{fig-wait} shows the experimental results when varying
{\wait} from $2$ to $6$ minutes, where the first two sub-figures are
trade-off graphs.
With a larger waiting time, the served rates of all the algorithms
increase because more driver candidates can satisfy the waiting time
constraint.
We observe that {\grp} consistently outperforms {\dsa} by
approximately $10$\% on served rate in Fig.~\ref{wait_servedRate},
while the average additional distance is increased slightly in
Fig.~\ref{wait_dis}, which validates our proposed goal to serve more
riders with less additional distance.
{\dsa} and {\dfp} maintain a similar served rate because both process
rider requests in order of the submission time, regardless of the
served rate. Moreover, it is reasonable that the additional distance of \dfp may not keep a consistent decreasing trend in Fig. \ref{wait_dis}.
Increasing \wait implies that more riders can share a trip and are also required to be served to their destinations. This could either increase or decrease the additional distance for the riders depending on multiple factors, such as request time, pickup, and drop off locations.
In term of matching time, we observe that {\dsa} grows rapidly among
all the algorithms in Fig.~\ref{wait_servedRate} because {\dsa}
enumerates all the possible positions to insert {\source} and {\des}
for a new rider, which requires $O(n!)$ time, where $n$ is the number
of points in the trip schedule.
With a larger {\wait}, more riders may share the vehicle such that
$n$ is larger.
In addition, the matching time of {\grp} is more than $10$ times
faster than {\gr}, which validates the efficiency of our proposed
pruning rules.
For the update time in Fig.~\ref{wait_update}, our proposed methods
are three orders-of-magnitude faster than {\dsa}.
Because {\dsa} must update two kinds of index structures: (1) each
driver maintains a kinetic tree~\cite{huang2014large} to record a
possible trip schedule for all unfinished requests assigned to each
vehicle, which takes $O(|D| \times n!)$~\cite{chen2018priceAndtime};
(2) a grid index to store road network information and vehicles that
are currently located or scheduled to enter each grid, which also has
an $O(|D|\times n!)$ update time
complexity~\cite{chen2018priceAndtime}.
In contrast, our methods only take $O(|D|\times n)$ to update the
trip schedule for each vehicle, and $O(D)$ to update the road network
index.
Thus, we find that {\dsa} cannot finish updating with the allocated
window of time ($10$ seconds), while our methods can.
Although {\grp} performs worse than {\dfp}, it can
have a higher served rate with less additional distance.

\begin{table}[t]
	\renewcommand{\arraystretch}{1.3}
	\caption{Parameter settings}
	\vspace{-3ex}
	\centering
		\begin{tabular}{lp{3.2cm}}
			\hline
			{\bfseries Parameter} & {\bf Setting}\\ 
			\hline
			Waiting time constraint $w$ (min) & 2, 3, 4, \textbf{5}, 6\\ \hline
			Detour time constraint \detour &  0.2, 0.4, \textbf{0.6}, 0.8, 1 \\ \hline
			Capacity of vehicle \ms & 2, 3, \textbf{4}, 5, 6\\ \hline
			Number of vehicles $|D|$ & \textbf{1000},1200,1400,1600,1800\\ \hline
			Size of the update timeslot \updatetimeinter(s) & 2, 5, \textbf{10}, 15, 20\\ \hline
			Number of partitions $\tau$ & \textbf{500}\\ \hline
			Travel speed \speed (km/hr) & 12, 24, 36, \textbf{48}, 60 \\ \hline
		\end{tabular}
	\label{table:parameter}
\end{table}

\noindent\textbf{Effect of the Capacity \ms.}
Fig.~\ref{fig-seat} illustrates the results when varying the seat
capacity from $2$ to $6$.
With a larger {\ms}, all the algorithms have an increased served
rate and lower additional distance because the vehicle can hold more
riders.
In term of matching time, {\df} and {\gr} are relatively stable.
When the capacity is small, most vehicles can be filtered by the
capacity constraint, thus the matching time is lower.
Our proposed methods have a similar update mechanism, and thus the
update time is dramatically lower than {\dsa}.


\noindent\textbf{Effect of the Number of Drivers $|D|$. }
Fig.~\ref{fig-driver} plots the results when varying the number of
drivers from $1{,}000$ to $1{,}800$.
{\grp} still maintains the highest served rate relative to the other
algorithms in Fig.~\ref{driver_servedRate}, while {\grp} can have
a similar additional distance to {\dsa} as shown in 
Fig.~\ref{driver_dis}. 
In term of matching time, {\dsa}, {\df} and {\gr} maintain a steady
growth trend as more drivers are available to serve a rider,
which results in an increase in the number of driver verifications.
Nevertheless, our proposed pruning rules can reduce the matching time
of {\df} and {\gr} by an order-of-magnitude in
Fig.~\ref{driver_servedRate}.

\noindent\textbf{Effect of the Travel Speed \speed. }
Fig.~\ref{fig-speed} describes the results of varying the travel
speed from $12$ to $60$ km/hr, which was also explored by 
Wang et al.~\cite{wang2016speed} on speed variations during peak hour.
We perform the experiments with different travel speeds in order to
simulate diverse traffic conditions in peak hours that have a direct
influence on the waiting time constraint.
Obviously, the served rate of all the algorithms increases when the
travel speed increases, because more rider requests can satisfy
the waiting time constraint.
When the travel speed is $12$ km/h, \gr (or \grp) can achieve a
$28.2\%$ improvement on served rate over \dsa from Fig.
\ref{speed_servedRate}, which corresponds to a slightly lower
additional distance from Fig.~\ref{speed_dis}.
This demonstrates that our methods can outperform the baseline
consistently even under serious traffic congestion scenarios.
From Fig.~\ref{speed_servedRate}, the matching time of \dsa keeps
increasing with larger travel speed because more rider requests need
to be considered for the assignment.
In contrast, our methods are more robust.

\noindent\textbf{Effect of the Detour Time Constraint \detour.}
Fig.~\ref{fig-detour} depicts the experimental results as 
{\detour} is varied from $0.2$ to $1$.
The served rate for all of the algorithms increases slightly with the
change of {\detour} in Fig.~\ref{detour_servedRate}, 
as the existing riders arranged in the trip schedule
allow a greater detour distance to support more riders.
Specifically, we can observe that {\grp} consistently has a higher 
served rate than {\dsa} or {\dfp}, while {\grp} degrades with the
additional distance, and maintains a lower average additional
distance with {\dsa} when $\detour=1$, as shown in
Fig.~\ref{detour_dis}.
In term of matching time, {\dsa} increases linearly in
Fig.~\ref{detour_servedRate} because more riders can share the trip,
which leads to more verified vehicles and longer matching time.

\noindent\textbf{Effect of the Update Time Window \updatetimeinter. }
Fig.~\ref{fig-window} shows the results of varying the size of
update time window from $2$ to $20$ seconds.
The served rates of {\dsa}, {\df}, and {\dfp} are insensitive to
the change of {\updatetimeinter}, because they process each rider
request in a first-come-first-serve manner.
{\grp} has an increased serving rate with a larger {\updatetimeinter},
because it can consider more rider requests within a time window and
generate a better result to both maximize the served rate
and minimize the additional distance. 
However, the served rates of all the algorithms drop when
$\updatetimeinter > 15$, because we move the vehicle every
{\updatetimeinter}, and the seat capacity becomes saturated in the
current timeslot.
This also results in decreased matching times for most algorithms
when $\updatetimeinter > 15$, while {\grp} increases slightly because it must still find the best possible matching.

\noindent \textbf{Summary of the experimental results.}
\dfp and \grp are improved versions of \df and \gr whose pruning power consistently reduces the number of candidates considered in each time window. Pruning is more essential in high demand scenarios. In our experiments, \grp is preferable when the served rate exceeds $30\%$, and \dfp is preferable when rider assignments must be executed in milliseconds.  
\begin{figure}[t]
	\centering
	\subfloat[Update time (The update time of \dsa is over $14$s, we omit it in the plot for better visualization)\label{detour_update}]{
		\includegraphics[width=4.23cm]{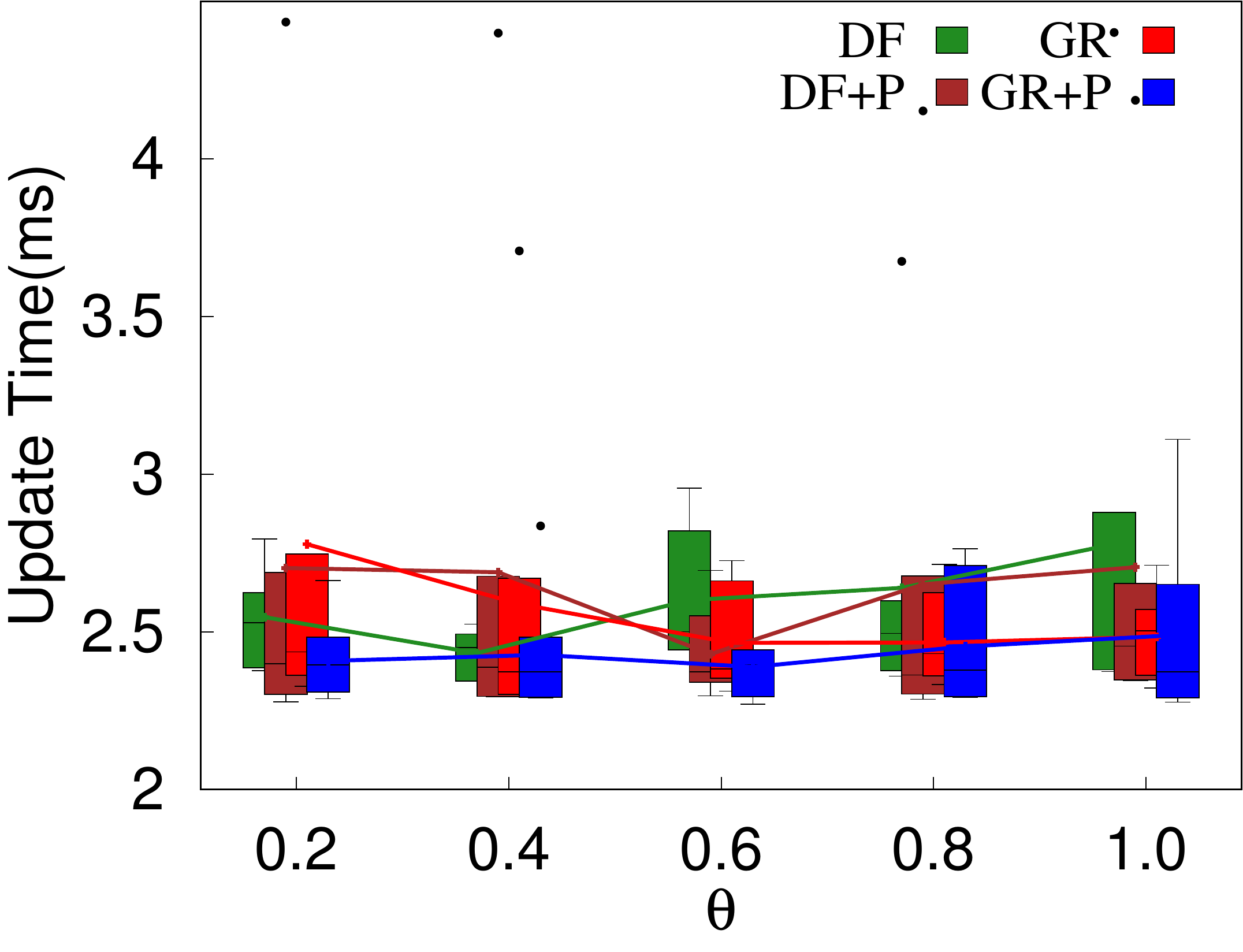}}
	\hspace{1px}
	\subfloat[Update time (The update time of \dsa is over $15$s, we omit it in the plot for better visualization)\label{window_update}]{
		\includegraphics[width=4.23cm]{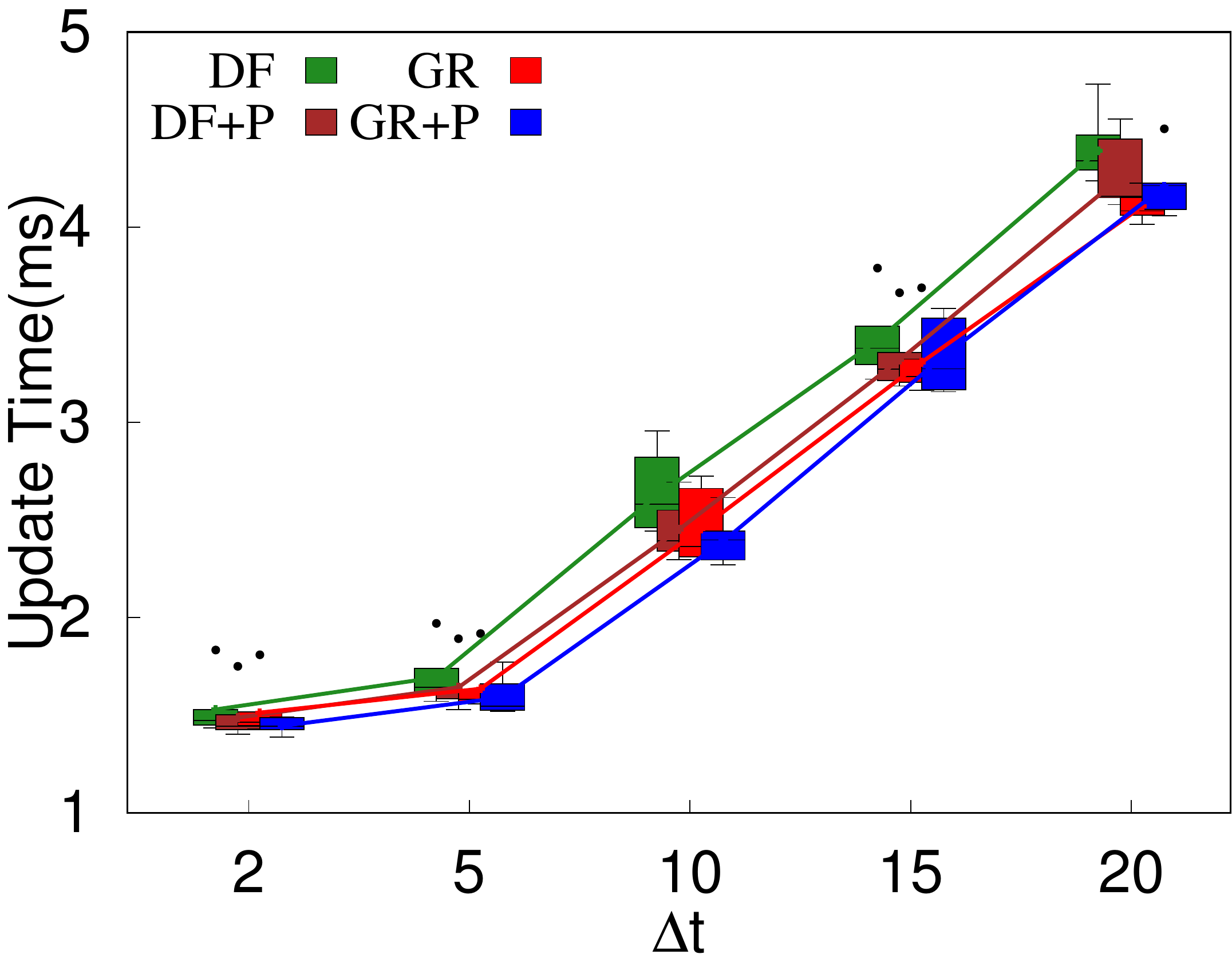}}
	\vspace{-1ex}
	\caption{The update time when varying {\detour} and {\updatetimeinter}}
	\label{fig-detourAndwindow-updateTime}  
	\vspace{-5px}
\end{figure}
\def\D{\leavevmode\hphantom{1}}

\begin{table*}[t]
  \renewcommand{\arraystretch}{1.3}
  \caption{The performance of DF+P and GR+P in a Peak Hour simulation as $|D|$ varies}
  \vspace{-8pt}
  \centering
  \begin{tabular}{cp{1cm}<{\centering}p{1cm}<{\centering}p{1.3cm}<{\centering}p{1.3cm}<{\centering}p{1cm}<{\centering}p{1cm}<{\centering}p{1cm}<{\centering}p{1cm}<{\centering}}
  \hline
  \multirow{2}{*}{\bfseries $|D|$} & \multicolumn{2}{c}{\bf Served Rate} & \multicolumn{2}{c}{\bf Additional Distance (km)} & \multicolumn{2}{c}{\bf Matching Time (ms)} & \multicolumn{2}{c}{\bf Update Time (ms)} \\ \cline{2-9}
  & \dfp               & \grp              & \dfp                     & \grp                     & \dfp                  & \grp                  & \dfp                 & \grp                 \\ \hline
  1000                 & 0.175            & 0.192           & 2.314                  & 1.806                  & 2.878               & \D3.962               & \D3.191              & \D3.495               \\ 
  3000                 & 0.422            & 0.507           & 2.294                  & 1.912                  & 10.543			     & \D12.375               & \D12.945              & \D13.141              \\ 
  5000                 & 0.601            & 0.711           & 2.270                   & 2.038                  & 12.177               & \D16.870                & 15.633             & 15.527             \\ 
  7000                 & 0.693             & 0.825           & 2.287                  & 2.130                  & 18.848               & \D34.195                & 21.049             & 28.226             \\ 
  9000                 & 0.770            & 0.884           & 2.275                  & 2.175                  & 25.875                & \D48.375              & 27.566              & 34.476              \\ \hline
  \end{tabular}
  \label{table:casestudy}
  \vspace{-8pt}
\end{table*}

\vspace{-7px}
\subsection{Case Study}\label{subsec:exp:result:case}
In order to show the scalability of our proposed
methods on both efficiency and effectiveness in peak travel periods
scenario, we adopt
{\dfp} and {\grp} to process $100,000$ orders in peak hours, where
there are $20$ rider requests per second on average.
We evaluate differing numbers of drivers in
Table~\ref{table:casestudy} to simulate a peak hour scenario, while
other parameter values are held constant.

We could observe that {\grp} can obtain a $18.3$\% higher served rate than
{\dfp} when $|D|=5000$, while requiring a smaller additional distance.
In term of efficiency, {\dfp} outperforms {\grp} because {\grp}
requires more rider insertion algorithm invocations to find the
minimal utility gain. 
Moreover, the update time of both {\dfp} and {\grp} can be finished
within the current timeslot.
\vspace{-10px}
\subsection{Performance Bound}\label{subsec:exp:result:performance}
In addition, we perform a performance bound
experiment to better observe the volatility of the algorithms
explored in this work.
As in prior work~\cite{pan2019ridesharing}, we use Simulated
Annealing (\sa) to find the best solution under each time window for
performance comparison.

Note that we have two objectives to achieve, thus we combine these
two objectives into one single function by minimizing the utility
function \utilitygr (i.e., Eq.~\ref{gr:utility}), where the smaller
the utility value is, and the better the method will be.
To be specific, we first generate an initial solution using {\grp}.
The initial solution could be chosen randomly without any
constraint violation, but using {\grp} achieves a better
result as verified in Section~\ref{subsec:exp:result:real}.
Next, we perform $P$ perturbations for a number of $T$
temperatures.
During each perturbation, we randomly select a rider request and then
reassign it to a different driver, where the constraints will be
further examined for the possible insertion.
If a smaller utility value is obtained, then we insert the rider and
update the driver's trip schedule.
Otherwise, if the utility value increases, then it is accepted
with probability $p = \exp (-\frac{\Delta U}{T})$, where
$\Delta U$ is the utility value difference before and after
reassignment, and $T$ is the current temperature.
The temperature is then reduced by a decay value $\delta*T$
proportional to the current temperature after every $P$
perturbations.
The algorithm terminates when the temperature reaches zero.
This alleviates the likelihood of converging to the same local optima
repeatedly.
If the temperature cool-down requires $t$ iterations, the time
complexity of {\sa} is $O(tPn^2 \shortestpathquery)$ including the
quadratic cost $O(n^2 \shortestpathquery)$ for each rider request
insertion as discussed in Section~\ref{sec:riderinsertion}.
We omit the time complexity of finding the initial solution which only
takes several milliseconds using {\grp} as it is not a factor
(quantified below).

For all experiments, we use the default parameter setting from Table
2.
The number of perturbations $P$ is $10{,}000$, the start temperature
$T$ is set to $5$, and the decay parameter $\delta = 0.001$.
We process $1,000$ orders in five successive time windows.
The experiments are repeated $10$ times and the average results are
reported.
The average matching time to process all the orders requires $28.3$
minutes ($1{,}698$ ms per rider request), while {\grp} only requires
$3.253$ seconds ($3.253$ ms per rider request).
For the effectiveness, the utility results of all methods are
presented in Table~\ref{table-pb}.
The effectiveness of {\df} and {\gr} are excluded since they
are the same as {\dfp} and {\grp}, respectively.
\vspace{-2ex}
\begin{table}[t]
	\centering
	\renewcommand{\arraystretch}{1.3}
	\caption{Utility Scores for all Methods over Five Successive Time Windows. }
	\vspace{-2ex}
	\begin{tabular}{cccccc}
		\hline
		& $1^{st}$ & $2^{nd}$ & $3^{rd}$ & $4^{th}$ & $5^{th}$ \\ \hline
		\dsa	&2.457	&2.223	&2.247	&2.305	&2.719      \\ \hline
		\dfp	&2.575	&2.580	&2.484	&2.537	&2.869      \\ \hline
		\grp	&2.399	&2.155	&2.023	&2.128	&2.670      \\ \hline
		\sa		&\textbf{2.157}	&\textbf{1.947}	&\textbf{1.889}	&\textbf{1.951}	&\textbf{2.281}		\\ \hline
	\end{tabular}
	\label{table-pb}
\end{table}

	\vspace{-5px}
\section{Conclusion} \label{sec:conclusion}
In this paper, we proposed a variant of the dynamic ridesharing
problem, which is a bilateral matching between a set of drivers and
riders, to achieve two optimization objectives: maximize
the served rate and minimize the total additional distance.
In our problem, we mainly focus on the peak hour case, where the
number of available drivers is insufficient to serve all of the rider
requests.
To speed up the matching process, we construct an index structure
based on a partitioned road network and compute the lower bound
distance between any two vertices on a road network.
Furthermore, we propose several pruning rules on top of the lower
bound distance estimation.
Two heuristic algorithms are devised to solve the bilateral matching
problem.
Finally, we carry out an experimental study on a large-scale real
dataset to show that our proposed algorithms have better efficiency
and effectiveness than the state-of-the-art method for dynamic
ridesharing.
In the future, we plan to investigate how
``fairness'' can be applied to the driver-and-rider matching problem
by defining a fair price mechanism for riders based on waiting time
tolerance.

	\vspace{-13px}
	\section*{Acknowledgments}
	This work was partially supported by ARC under Grants DP170102726, DP170102231, DP180102050, and DP200102611, and the National Natural Science Foundation of China (NSFC) under Grants 61728204, and 91646204. Zhifeng Bao and J. Shane Culpepper are the recipients of Google Faculty Award. 
	\vspace{-12px}
	\bibliographystyle{IEEEtran}
	\bibliography{ref_all}

\begin{thebibliography}{10}
\providecommand{\url}[1]{#1}
\csname url@samestyle\endcsname
\providecommand{\newblock}{\relax}
\providecommand{\bibinfo}[2]{#2}
\providecommand{\BIBentrySTDinterwordspacing}{\spaceskip=0pt\relax}
\providecommand{\BIBentryALTinterwordstretchfactor}{4}
\providecommand{\BIBentryALTinterwordspacing}{\spaceskip=\fontdimen2\font plus
\BIBentryALTinterwordstretchfactor\fontdimen3\font minus
  \fontdimen4\font\relax}
\providecommand{\BIBforeignlanguage}[2]{{%
\expandafter\ifx\csname l@#1\endcsname\relax
\typeout{** WARNING: IEEEtran.bst: No hyphenation pattern has been}%
\typeout{** loaded for the language `#1'. Using the pattern for}%
\typeout{** the default language instead.}%
\else
\language=\csname l@#1\endcsname
\fi
#2}}
\providecommand{\BIBdecl}{\relax}
\BIBdecl

\bibitem{didi}
\BIBentryALTinterwordspacing
Didichuxing. [Online]. Available: \url{https://www.didiglobal.com/}
\BIBentrySTDinterwordspacing

\bibitem{uber}
\BIBentryALTinterwordspacing
Uberpool. [Online]. Available: \url{https://www.uber.com}
\BIBentrySTDinterwordspacing

\bibitem{zhang2017taxi}
L.~Zhang, T.~Hu, Y.~Min, G.~Wu, J.~Zhang, P.~Feng, P.~Gong, and J.~Ye, ``A taxi
  order dispatch model based on combinatorial optimization,'' in \emph{Proc.
  SIGKDD}, 2017, pp. 2151--2159.

\bibitem{teubner2015economics}
T.~Teubner and C.~M. Flath, ``The economics of multi-hop ride sharing,''
  \emph{Business \& Information Systems Engineering}, vol.~57, no.~5, pp.
  311--324, 2015.

\bibitem{alonso2017demand}
J.~Alonso-Mora, S.~Samaranayake, A.~Wallar, E.~Frazzoli, and D.~Rus,
  ``On-demand high-capacity ride-sharing via dynamic trip-vehicle assignment,''
  \emph{Proc. National Academy of Sciences}, vol. 114, no.~3, pp. 462--467,
  2017.

\bibitem{bitre}
\BIBentryALTinterwordspacing
Bureau of infrastructure, transport and regional economics. [Online].
  Available: \url{https://bitre.gov.au/publications/2015/yearbook_2015.aspx/}
\BIBentrySTDinterwordspacing

\bibitem{jacobson2009fuel}
S.~H. Jacobson and D.~M. King, ``Fuel saving and ridesharing in the us:
  Motivations, limitations, and opportunities,'' \emph{Transportation Research
  Part D: Transport and Environment}, vol.~14, no.~1, pp. 14--21, 2009.

\bibitem{didireport}
\BIBentryALTinterwordspacing
{Transforming the Mobility Industry: The Story of Didi}. [Online]. Available:
  \url{http://d20forum.org/wp-content/uploads/2017/12/11.50_Didi_Transforming-mobility_FINAL.pdf}
\BIBentrySTDinterwordspacing

\bibitem{downs2005still}
\BIBentryALTinterwordspacing
A.~Downs, \emph{Still stuck in traffic: coping with peak-hour traffic
  congestion}, 2005. [Online]. Available:
  \url{http://www.jstor.org/stable/10.7864/j.ctt1vjqprt}
\BIBentrySTDinterwordspacing

\bibitem{zheng2018order}
L.~Zheng, L.~Chen, and J.~Ye, ``Order dispatch in price-aware ridesharing,''
  \emph{PVLDB}, vol.~11, no.~8, pp. 853--865, 2018.

\bibitem{chen2018priceAndtime}
L.~Chen, Q.~Zhong, X.~Xiao, Y.~Gao, P.~Jin, and C.~S. Jensen,
  ``Price-and-time-aware dynamic ridesharing,'' in \emph{Proc. ICDE}, 2018, pp.
  1061--1072.

\bibitem{ta2018efficient}
N.~Ta, G.~Li, T.~Zhao, J.~Feng, H.~Ma, and Z.~Gong, ``An efficient ride-sharing
  framework for maximizing shared route,'' \emph{TKDE}, vol.~30, no.~2, pp.
  219--233, 2018.

\bibitem{ma2013t}
S.~Ma, Y.~Zheng, and O.~Wolfson, ``T-share: A large-scale dynamic taxi
  ridesharing service,'' in \emph{Proc. ICDE}, 2013, pp. 410--421.

\bibitem{huang2014large}
Y.~Huang, F.~Bastani, R.~Jin, and X.~S. Wang, ``Large scale real-time
  ridesharing with service guarantee on road networks,'' \emph{PVLDB}, vol.~7,
  no.~14, pp. 2017--2028, 2014.

\bibitem{cheng2017utility}
P.~Cheng, H.~Xin, and L.~Chen, ``Utility-aware ridesharing on road networks,''
  in \emph{Proc. SIGMOD}, 2017, pp. 1197--1210.

\bibitem{xing2009smize}
X.~Xing, T.~Warden, T.~Nicolai, and O.~Herzog, ``Smize: a spontaneous
  ride-sharing system for individual urban transit,'' in \emph{Proc. MATES},
  2009, pp. 165--176.

\bibitem{fu2017top}
X.~Fu, J.~Huang, H.~Lu, J.~Xu, and Y.~Li, ``Top-k taxi recommendation in
  realtime social-aware ridesharing services,'' in \emph{Proc. SSTD}, 2017, pp.
  221--241.

\bibitem{tong2018unified}
Y.~Tong, Y.~Zeng, Z.~Zhou, L.~Chen, J.~Ye, and K.~Xu, ``A unified approach to
  route planning for shared mobility,'' \emph{PVLDB}, vol.~11, no.~11, pp.
  1633--1646, 2018.

\bibitem{cao2015sharek}
B.~Cao, L.~Alarabi, M.~F. Mokbel, and A.~Basalamah, ``Sharek: A scalable
  dynamic ride sharing system,'' in \emph{Proc. MDM}, vol.~1, 2015, pp. 4--13.

\bibitem{ma2013analysis}
S.~Ma and O.~Wolfson, ``Analysis and evaluation of the slugging form of
  ridesharing,'' in \emph{Proc. of SIGSPATIAL}, 2013, pp. 64--73.

\bibitem{lin2012research}
Y.~Lin, W.~Li, F.~Qiu, and H.~Xu, ``Research on optimization of vehicle routing
  problem for ride-sharing taxi,'' \emph{Procedia-Social and Behavioral
  Sciences}, vol.~43, pp. 494--502, 2012.

\bibitem{yan2011optimization}
S.~Yan and C.-Y. Chen, ``An optimization model and a solution algorithm for the
  many-to-many car pooling problem,'' \emph{Annals of Operations Research},
  vol. 191, no.~1, pp. 37--71, 2011.

\bibitem{wang2016pickup}
X.~Wang, M.~Dessouky, and F.~Ordonez, ``A pickup and delivery problem for
  ridesharing considering congestion,'' \emph{Transportation letters}, vol.~8,
  no.~5, pp. 259--269, 2016.

\bibitem{hall1997dynamic}
R.~W. Hall and A.~Qureshi, ``Dynamic ride-sharing: Theory and practice,''
  \emph{Journal of Transportation Engineering}, vol. 123, no.~4, pp. 308--315,
  1997.

\bibitem{tao2007dynamic}
C.-C. Tao, ``Dynamic taxi-sharing service using intelligent transportation
  system technologies,'' in \emph{Proc. of WiCOM}, 2007, pp. 3209--3212.

\bibitem{chan2012ridesharing}
N.~D. Chan and S.~A. Shaheen, ``Ridesharing in north america: Past, present,
  and future,'' \emph{Transport Reviews}, vol.~32, no.~1, pp. 93--112, 2012.

\bibitem{morency2007ambivalence}
C.~Morency, ``The ambivalence of ridesharing,'' \emph{Transportation}, vol.~34,
  no.~2, pp. 239--253, 2007.

\bibitem{cordeau2007dial}
J.-F. Cordeau and G.~Laporte, ``The dial-a-ride problem: models and
  algorithms,'' \emph{Annals of operations research}, vol. 153, no.~1, pp.
  29--46, 2007.

\bibitem{psaraftis1980dynamic}
H.~N. Psaraftis, ``A dynamic programming solution to the single vehicle
  many-to-many immediate request dial-a-ride problem,'' \emph{Transportation
  Science}, vol.~14, no.~2, pp. 130--154, 1980.

\bibitem{bei2018algorithms}
X.~Bei and S.~Zhang, ``Algorithms for trip-vehicle assignment in
  ride-sharing,'' in \emph{Proc. AAAI}, 2018.

\bibitem{agatz2012optimization}
N.~Agatz, A.~Erera, M.~Savelsbergh, and X.~Wang, ``Optimization for dynamic
  ride-sharing: A review,'' \emph{European Journal of Operational Research},
  vol. 223, no.~2, pp. 295--303, 2012.

\bibitem{furuhata2013ridesharing}
M.~Furuhata, M.~Dessouky, F.~Ord{\'o}{\~n}ez, M.-E. Brunet, X.~Wang, and
  S.~Koenig, ``Ridesharing: The state-of-the-art and future directions,''
  \emph{Transportation Research Part B: Methodological}, vol.~57, pp. 28--46,
  2013.

\bibitem{agatz2011dynamic}
N.~A. Agatz, A.~L. Erera, M.~W. Savelsbergh, and X.~Wang, ``Dynamic
  ride-sharing: A simulation study in metro atlanta,'' \emph{Transportation
  Research Part B: Methodological}, vol.~45, no.~9, pp. 1450--1464, 2011.

\bibitem{kleiner2011mechanism}
A.~Kleiner, B.~Nebel, and V.~A. Ziparo, ``A mechanism for dynamic ride sharing
  based on parallel auctions,'' in \emph{Proc. IJCAI}, vol.~11, 2011, pp.
  266--272.

\bibitem{ma2015real}
S.~Ma, Y.~Zheng, and O.~Wolfson, ``Real-time city-scale taxi ridesharing,''
  \emph{TKDE}, vol.~27, no.~7, pp. 1782--1795, 2015.

\bibitem{duan2016real}
X.~Duan, C.~Jin, X.~Wang, A.~Zhou, and K.~Yue, ``Real-time personalized
  taxi-sharing,'' in \emph{Proc. DASFAA}, 2016, pp. 451--465.

\bibitem{xu2019efficient}
Y.~Xu, Y.~Tong, Y.~Shi, Q.~Tao, K.~Xu, and W.~Li, ``An efficient insertion
  operator in dynamic ridesharing services,'' in \emph{Proc. ICDE}, 2019, pp.
  1022--1033.

\bibitem{tong2016onlineminimum}
Y.~Tong, J.~She, B.~Ding, L.~Chen, T.~Wo, and K.~Xu, ``Online minimum matching
  in real-time spatial data: experiments and analysis,'' \emph{PVLDB}, vol.~9,
  no.~12, pp. 1053--1064, 2016.

\bibitem{tong2016onlinemobile}
Y.~Tong, J.~She, B.~Ding, L.~Wang, and L.~Chen, ``Online mobile micro-task
  allocation in spatial crowdsourcing,'' in \emph{Proc. ICDE}, 2016, pp.
  49--60.

\bibitem{abraham2011hub}
I.~Abraham, D.~Delling, A.~V. Goldberg, and R.~F. Werneck, ``A hub-based
  labeling algorithm for shortest paths in road networks,'' in \emph{Proc.
  SEA}, 2011, pp. 230--241.

\bibitem{karypis1998fast}
G.~Karypis and V.~Kumar, ``A fast and high quality multilevel scheme for
  partitioning irregular graphs,'' \emph{SIAM Journal on Scientific Computing},
  vol.~20, no.~1, pp. 359--392, 1998.

\bibitem{li2017experimental}
Y.~Li, M.~L. Yiu, N.~M. Kou \emph{et~al.}, ``An experimental study on hub
  labeling based shortest path algorithms,'' \emph{PVLDB}, vol.~11, no.~4, pp.
  445--457, 2017.

\bibitem{wang2016speed}
X.~Wang, T.~Fan, W.~Li, R.~Yu, D.~Bullock, B.~Wu, and P.~Tremont, ``Speed
  variation during peak and off-peak hours on urban arterials in shanghai,''
  \emph{Transportation Research Part C: Emerging Technologies}, vol.~67, pp.
  84--94, 2016.

\bibitem{pan2019ridesharing}
J.~J. Pan, G.~Li, and J.~Hu, ``Ridesharing: simulator, benchmark, and
  evaluation,'' \emph{PVLDB}, vol.~12, no.~10, pp. 1085--1098, 2019.

\end{thebibliography}
	\vspace{-4em}
	\begin{IEEEbiography}[{\includegraphics[width=1in,clip,keepaspectratio]{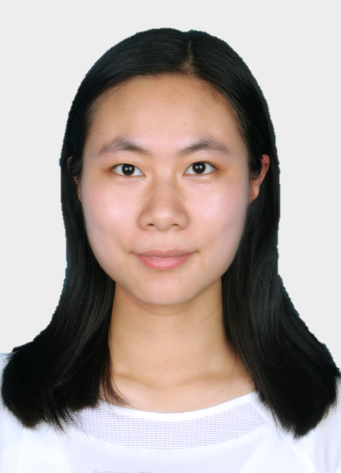}}]{Hui Luo}
		received the M.E. degree in computer science from Wuhan University,
		Hubei, China, in 2017. She is currently pursuing the
		Ph.D. degree at Computer Science and Information Technology, RMIT University, Melbourne, Australia. Her current research interests include data mining, intelligent transportation, and machine learning. 
	\end{IEEEbiography}
	\vspace{-4.5em}
	\begin{IEEEbiography}[{\includegraphics[width=1in,height=1.25in,clip,keepaspectratio]{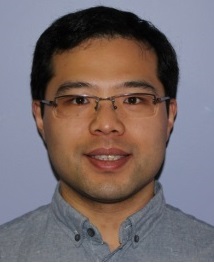}}]{Zhifeng Bao}
		received the Ph.D. degree in computer science from the National University of Singapore in 2011 as the winner of the Best PhD
		Thesis in school of computing. He is currently a
		senior lecturer with the RMIT University and leads
		the big data research group at RMIT. He is also
		an Honorary Fellow with University of Melbourne
		in Australia. His current research interests include
		data usability, spatial database, data integration,
		and data visualization. 
	\end{IEEEbiography}
	\vspace{-5.5em}
	\begin{IEEEbiography}[{\includegraphics[width=1in,height=1.25in,clip,keepaspectratio]{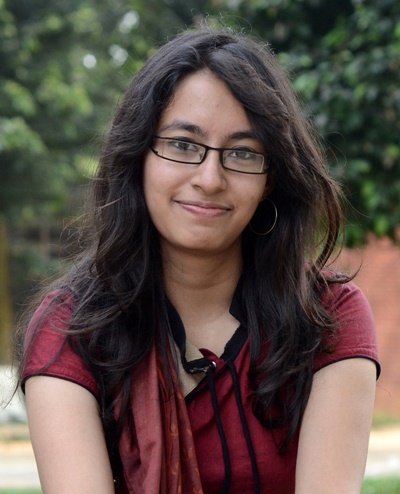}}]{Farhana M. Choudhury}
		received the Ph.D. degree in computer science from RMIT University in 2017. She is currently a lecturer at The University of Melbourne. Her current research interests include
		spatial databases, data visualization, trajectory queries, and applying machine learning techniques to solve spatial problems. 
	\end{IEEEbiography}
	\vspace{-5em}
	\begin{IEEEbiography}[{\includegraphics[width=1in,height=1.25in,clip,keepaspectratio]{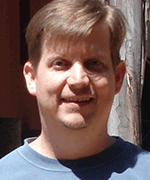}}]{J. Shane Culpepper}received the Ph.D. degree at The
	University of Melbourne in 2008. He is currently a Vice-Chancellor's Principal Research Fellow and Director of the Centre for Information Discovery and Data Analytics at RMIT University. His research interests include information retrieval, machine learning, text indexing, system evaluation, algorithm engineering, and spatial computing.
	\end{IEEEbiography}
\end{document}